\newif\iffull
\newif\ifnotes
\newif\ifanonymous
\newif\ifcameraready

\fulltrue
\notesfalse

\ifcameraready
  \PassOptionsToClass{nostandalone}{llncsstandalone}
\else
  \PassOptionsToClass{adjustmargins}{llncsstandalone}
\fi
\documentclass[a4paper,orivec,oribibl]{llncsstandalone}

\usepackage[save]{silence}
\usepackage[T1]{fontenc}
\usepackage{lmodern}
\usepackage[english]{babel}
\usepackage{csquotes}
\usepackage[babel,final]{microtype}
\usepackage[dvipsnames,svgnames,x11names]{xcolor}
\usepackage{amsmath,amssymb,mathtools}
\usepackage{booktabs}
\usepackage{graphicx}
\usepackage{enumitem}
\usepackage{xspace}
\usepackage[backend=biber,style=alphabetic,maxnames=3,minnames=3,maxbibnames=20,minbibnames=20,date=year]{biblatex}
\PassOptionsToPackage{pdfusetitle}{hyperref}
\usepackage{hyperref}
\hypersetup{hypertexnames=false,colorlinks=true,allcolors=blue,bookmarksopen=true}
\usepackage{bookmark}
\usepackage[capitalize,nameinlink]{cleveref}
\crefname{equation}{}{}
\Crefname{equation}{Equation}{Equations}
\crefname{definition}{Definition}{Definitions}
\Crefname{definition}{Definition}{Definitions}
\crefname{proposition}{Proposition}{Propositions}
\Crefname{proposition}{Proposition}{Propositions}
\crefname{lemma}{Lemma}{Lemmas}
\Crefname{lemma}{Lemma}{Lemmas}
\crefname{theorem}{Theorem}{Theorems}
\Crefname{theorem}{Theorem}{Theorems}
\crefname{corollary}{Corollary}{Corollaries}
\Crefname{corollary}{Corollary}{Corollaries}

\newcommand{\CKKS}{\ensuremath{\mathsf{CKKS}}\xspace}
\newcommand{\Rot}{\mathsf{Rot}}
\newcommand{\rev}{\operatorname{rev}}
\newcommand{\Z}{\mathbb{Z}}
\newcommand{\nuTwo}{\nu_2}
\newcommand{\defeq}{\coloneqq}

\title{\texorpdfstring{%
Rotation-Optimal Noncommutative Prefix Scans\\
in Bit-Reversed Homomorphic Layouts%
}{Rotation-Optimal Noncommutative Prefix Scans in Bit-Reversed Homomorphic Layouts}}

\author{\texorpdfstring{%
Anis Bkakria\inst{1} \and
Madicke-Diadji MBODJ\inst{1,2} \and
Mawloud Omar\inst{2} \and
Reda Yaich\inst{1}%
}{Anis Bkakria, Madicke-Diadji MBODJ, Mawloud Omar, and Reda Yaich}}
\institute{IRT SystemX, France \and UBS, France}

\begin{document}
\maketitle

\begin{abstract}
Packed homomorphic encryption evaluates slotwise operations in parallel, but nonlocal communication is realized by cyclic rotations whose cost depends on the physical slot layout. We study ordered prefix computation on $n=2^m$ elements of an associative, possibly noncommutative monoid stored in bit-reversed order. A direct transported-predecessor scan uses $m\cdot(m+1)/2$ rotations because one logical shift decomposes into several cyclic displacement classes. We introduce a replicated-aggregate invariant in which every slot of an aligned logical block stores the same complete block aggregate. Semantic replication makes the copies interchangeable: at each level, one global cyclic rotation supplies every slot with a valid aggregate of its sibling block, even though it need not reach the exact logical partner. The resulting inclusive or exclusive scan uses $m$ rotations, monoid depth $m$, two live state vectors, and at most $2\cdot m-1$ packed monoid compositions.

In a model where all non-routing operations are slotwise and every cyclic rotation invocation is counted, both bounds are exact: $D^\star(m)=R^\star(m)=m$. Equality is rigid---the $m$ rotation offsets contain exactly one representative of every $2$-adic valuation $0,\ldots,m-1$. With at most $K$ directly keyed offsets, we prove a product lower bound on online rotation calls and an exact frontier $K\cdot(2^{m/K}-1)$ whenever $K$ divides $m$. We instantiate the exclusive scan for radix carry and borrow in bit-reversed \CKKS slots, avoiding both layout restoration and a final logical-predecessor shift. In our Lattigo implementation at $m=7$, the replicated scan reduces the direct bit-reversed baseline from $28$ to $7$ rotations, lowers evaluation-key storage by $70.0\%$, lowers peak heap usage by $63.9\%$, and improves isolated scan latency by $19.9\%$. In a depth-$5$ downstream pipeline, retaining six additional modulus levels avoids one bootstrap and yields a mean paired speedup of $4.31\times$ with a $95\%$ confidence interval of $[3.69,4.92]$.

\end{abstract}


\section{Introduction}\label{sec:introduction}

\paragraph{Packed computation and layout-sensitive communication.}
Modern lattice-based homomorphic-encryption schemes expose a SIMD abstraction: a ciphertext encrypts a vector of slots, and additions and multiplications act slotwise on the entire vector. In \CKKS, this abstraction supports approximate arithmetic over packed complex values and has become a standard foundation for encrypted numerical computation~\cite{CheonKKS17}. Slotwise parallelism, however, does not make communication free. Moving encrypted values between slots requires Galois automorphisms, commonly exposed as cyclic rotations followed by key switching. The number of logical dependencies in an algorithm and the number of encrypted rotations needed to realize them are therefore different resources.

The distinction becomes particularly sharp when data is kept in a public structured layout. Fast transforms, packing conversions, and transposed representations may leave logical indices in bit-reversed or related orders. A classical prefix network is normally analyzed by its gate count, depth, and fanout~\cite{LadnerFischer80,Harris03}. Those measures do not determine the encrypted routing cost: one rotation realizes the same cyclic displacement at all slots, so many logical edges can share a rotation, while a single logical predecessor relation may split into several displacement classes after conjugation by the layout. The relevant question is consequently not only how many prefix gates are needed, but how many \emph{global cyclic translations} are needed when the layout must be preserved.

\paragraph{Motivating application: exact carry in radix \CKKS.}
Approximate arithmetic does not directly provide a canonical representation of large integers. Radix-based approaches store an integer as encrypted digits and periodically propagate carry or borrow to return the digit vector to a unique range. Recent work of Cha, Park, and Lee gives a logarithmic-round exact-carry mechanism for radix \CKKS and uses it to support large-integer comparison and modular arithmetic~\cite{ChaPL26}. Carry propagation is an ordered prefix computation: every digit defines a transition on the incoming carry, and the carry entering position $i$ is obtained by composing the transitions of all less significant digits.

When the digits are already stored in bit-reversed order, two natural realizations incur additional routing. One may restore natural order, run a standard scan, and possibly restore the original layout. Alternatively, one may transport the usual recursive-doubling predecessor scan through bit reversal. For $n=2^m$ digits, the latter has $m$ composition levels, but its logical shift by $2^t$ decomposes into $m-t$ cyclic displacement classes. Its rotation count is therefore
\begin{equation}
  R_{\mathrm{direct}}(m)
  = \sum_{t=0}^{m-1}(m-t)
  = \frac{m\cdot(m+1)}{2}.
  \label{eq:direct-triangular}
\end{equation}
This count is exact for that transported-predecessor topology. It is not, as we show, an inherent cost of prefix computation in a bit-reversed layout.

\paragraph{The key idea: route semantic copies, not exact partners.}
The obstacle in the direct scan is overly specific routing. At a level that combines two adjacent logical blocks, each destination requests the value held by one precise logical predecessor. Under bit reversal, these requests occupy several cyclic diagonals. Our construction changes the invariant rather than searching for a better sequence of predecessor shifts.

For every aligned logical block, we maintain its complete monoid aggregate \emph{replicated at every slot belonging to the block}. We also maintain, in a second state vector, the ordered prefix local to that block. At level $d$, the two children of every length-$2^{d+1}$ parent occupy the two physical cosets separated by
\begin{equation}
  s_d = 2^{m-d-1}.
  \label{eq:stage-rotation}
\end{equation}
A rotation by $s_d$ does not necessarily map a slot to its exact XOR partner. It does something sufficient and cheaper: it maps the slot to \emph{some} position in the sibling child. Because the child aggregate is replicated, every such position contains the same semantic value. One rotation consequently exchanges valid sibling aggregates for all blocks and both child orientations at once. Public masks select the operand order, which preserves correctness for noncommutative monoids.

This semantic-copy observation reduces the routing cost from the triangular count in \cref{eq:direct-triangular} to one rotation per level. The price is explicit: the construction keeps two live monoid states and performs one aggregate update and one prefix update at each nonfinal level. The result is thus a multi-resource trade-off rather than a claim that rotations are the only meaningful cost.

\paragraph{Main theorem.}
Let $n=2^m$, and place logical input $x_i$ at physical slot $\rev_m(i)$. We consider acyclic, data-independent packed circuits in which masks, copying, and all other local operations are slotwise, while each invocation of a cyclic rotation is counted. Helper ciphertexts and recomputation are unrestricted, but an uncounted dense linear transform, arbitrary slot permutation, or bootstrapping transform is not available as hidden routing.

In this model, the replicated scan computes every inclusive or exclusive ordered prefix using
\begin{equation}
  D=m,
  \qquad
  R=m,
  \qquad
  T\leq 2\cdot m-1,
  \label{eq:main-cost}
\end{equation}
where $D$ is binary monoid depth, $R$ is the number of cyclic-rotation invocations, and $T$ is the number of packed monoid compositions. Both $D$ and $R$ are globally optimal:
\begin{equation}
  \boxed{D^\star(m)=R^\star(m)=m.}
  \label{eq:main-optimum}
\end{equation}
The depth lower bound follows from binary fan-in. For rotations, fix one output depending on all $2^m$ inputs. Along a dependency path, every one of the $R$ rotation calls is either traversed or not, so the source displacement is a subset sum of the invoked offsets. At most $2^R$ displacements can reach the output; hence $2^R\geq 2^m$ and $R\geq m$. This argument permits arbitrary helper width, branching, copying, masks, local nonlinear operations, and recomputation.

\paragraph{Structure at equality.}
Rotation optimality is rigid. Suppose exactly $m$ rotations with offsets $\delta_0,\ldots,\delta_{m-1}\in\Z_{2^m}$ suffice for one output to depend on every slot. The $2^m$ subset sums of these offsets must then be distinct. A roots-of-unity argument applied to
\begin{equation}
  \prod_{j=0}^{m-1}(1+X^{\delta_j})
  = 1+X+\cdots+X^{2^m-1}
  \pmod{X^{2^m}-1}
\end{equation}
shows that, after reordering,
\begin{equation}
  \nuTwo(\delta_j)=j
  \qquad\text{for }0\leq j<m.
  \label{eq:rigidity-intro}
\end{equation}
Thus every optimal offset family forms a complete $2$-adic basis: one offset appears at each valuation scale. This characterizes the rotation profile of every equality case, but it does not claim that the complete circuit topology is unique.

\paragraph{Rotation keys versus online calls.}
A rotation invocation also requires evaluation-key material for its offset, unless the rotation is synthesized from other directly supported offsets. This creates a time--memory trade-off studied in generic key-management work~\cite{LeeLKNo23}. Our result gives a program-specific lower bound. If at most $K$ directly keyed offsets are invoked $c_1,\ldots,c_K$ times, then one dependency path can realize at most
\begin{equation}
  \prod_{j=1}^{K}(c_j+1)
\end{equation}
coefficient vectors. Reaching all $2^m$ source displacements therefore requires
\begin{equation}
  \prod_{j=1}^{K}(c_j+1)\geq 2^m.
  \label{eq:key-product-intro}
\end{equation}
When $K$ divides $m$, this lower bound is attained by grouping consecutive $2$-adic scales and directly keying the smallest offset of each group. The exact frontier is
\begin{equation}
  R^\star(m,K)=K\cdot\bigl(2^{m/K}-1\bigr)
  \qquad(K\mid m).
  \label{eq:key-frontier-intro}
\end{equation}
For general $K$, we give an exact integer counting lower bound and a balanced dyadic upper bound; they coincide for all divisible cases and for additional nondivisible parameter pairs. We do not claim a closed formula when the two bounds differ.

\paragraph{Exclusive carry and borrow without a final shift.}
For radix carry, an inclusive prefix at digit $i$ naturally returns the outgoing carry $c_{i+1}$. The incoming carry $c_i$ would then have to be transported from logical position $i-1$, which is itself layout sensitive. We instead compute exclusive prefixes. Let $S_i$ be the local carry-transition state and let $E_i$ be the composition of $S_0,\ldots,S_{i-1}$. Then
\begin{equation}
  c_i = F_{E_i}(0),
  \qquad
  c_{i+1}=F_{S_i}(c_i),
  \qquad
  d_i=z_i+c_i-B\cdot c_{i+1}.
  \label{eq:exclusive-carry-intro}
\end{equation}
The exclusive scan outputs the prefix transition $E_i$; the carry-in, carry-out, and corrected digit are then evaluated slotwise. No final logical-predecessor shift or other cross-slot communication is required. The same exclusive-prefix mechanism gives borrow propagation and thereby supports comparison and conditional subtraction. With the fully occupied digit-major interleaving
\begin{equation}
  \Phi(r,i)=r+g\cdot\rev_m(i),
\end{equation}
all $g$ packed integers execute level $d$ with one physical rotation by $g\cdot 2^{m-d-1}$, so batching does not multiply the number of scan rotations.

\paragraph{Contributions.}
The paper makes five contributions.
\begin{enumerate}[leftmargin=*,label=\textbf{\arabic*.},itemsep=0.35em,topsep=0.35em]
  \item \textbf{A packed-rotation model and exact lower bounds.}
  We formalize ordered, possibly noncommutative prefix computation with counted cyclic rotations and prove the global bounds $D,R\geq m$.

  \item \textbf{A rotation-optimal replicated scan.}
  We give an inclusive and exclusive bit-reversed scan attaining $D=R=m$ with two live states and at most $2\cdot m-1$ packed monoid compositions. The central sibling-copy lemma explains why one directed cyclic rotation realizes both sides of a hypercube exchange.

  \item \textbf{$2$-adic rigidity and a rotation-key frontier.}
  We characterize the valuation profile of every $m$-rotation equality case, prove the direct-key product lower bound in \cref{eq:key-product-intro}, and obtain the exact frontier in \cref{eq:key-frontier-intro} when $K\mid m$.

  \item \textbf{Layout-native radix carry and borrow.}
  We instantiate the exclusive scan with ordered carry transitions, avoid both natural-order restoration and the final predecessor shift, and prove the corresponding segmented digit-major packing rule.

  \item \textbf{Implementation and end-to-end evaluation.}
  At $m=7$, the replicated implementation reduces direct bit-reversed routing from $28$ to $7$ rotations, lowers evaluation-key storage by $70.0\%$, lowers peak heap usage by $63.9\%$, and improves isolated scan latency by $19.9\%$. In a depth-$5$ downstream pipeline, preserving six additional modulus levels avoids one bootstrap and yields a mean paired speedup of $4.31\times$ with a $95\%$ confidence interval of $[3.69,4.92]$.
\end{enumerate}

\paragraph{Comparison with the closest lines of work.}
Classical prefix networks establish logarithmic-depth constructions and size--depth trade-offs for associative operators~\cite{LadnerFischer80,Harris03}. Our scan uses the same broad parallel-prefix principle, but studies a different cost model: a charged operation is a global cyclic translation shared across slots, not an individual network edge. The new ingredient is the semantic-replication invariant that makes an inexact physical partner sufficient under bit reversal, together with a matching lower bound in the packed-rotation model.

Cha, Park, and Lee already establish that radix-\CKKS carry can be evaluated in logarithmic rounds and provide an optimized symbolic carry algebra~\cite{ChaPL26}. We do not claim the first logarithmic carry circuit. Our contribution is orthogonal: we preserve a bit-reversed layout, attain the exact minimum number of scan rotations in that layout, characterize all equality profiles, and formulate carry exclusively so that no final logical shift is needed. Generic rotation-key work seeks to reduce transmitted or stored evaluation keys across broad workloads~\cite{LeeLKNo23}; our key frontier instead lower-bounds the online calls required by this specific dependency pattern.

\paragraph{Scope and limitations.}
The exact rotation theorem applies to one complete cyclic domain of $2^m$ active positions and to fully occupied equal-length interleavings for which each scan rotation preserves the object residue. Contiguous object-major segments and partially filled ciphertexts may require padding or additional masked routing. The lower bound counts online rotation invocations; it does not equate them with wall-clock latency, distinct Galois keys, hoisted key-switch decompositions, or total memory traffic. It also excludes uncounted dense transforms and treats the prefix/correction layer separately from the upstream encrypted digit-state classifier. The implementation therefore evaluates the layout-preserving scan and correction layer once encrypted carry-state symbols are available.

\paragraph{Organization.}
The remainder of the paper develops the packed-rotation model, proves the depth and rotation lower bounds and equality rigidity, presents and verifies the replicated exclusive scan, derives the rotation-key frontier, and then specializes the construction to radix-\CKKS carry, borrow, and segmented packing. The implementation and evaluation section then accounts for ciphertext products, levels, key material, precision, memory, and end-to-end performance against the direct and layout-restoring baselines, before the paper closes with related work and conclusions.

\section{Technical Overview}\label{sec:overview}

This section gives the construction and proof strategy before introducing the
full circuit model.  The central point is that a packed rotation need not
reproduce every logical edge of a classical prefix network.  It is sufficient
to route a semantically interchangeable copy of the information required at
each destination.

\subsection{Why transporting the predecessor scan is expensive}
\label{sec:overview-direct}

Consider eight logical inputs.  Three-bit reversal places them in physical
slot order
\begin{equation}
  0,4,2,6,1,5,3,7.
  \label{eq:overview-bitrev-eight}
\end{equation}
A natural recursive-doubling scan combines logical predecessors at distances
$1$, $2$, and $4$.  After conjugation by bit reversal, however, one logical
distance need not correspond to one cyclic displacement.  Under the rotation
convention used in the formal model, the three rounds of the eight-slot
instance require displacement sets
\begin{equation}
  \Delta_0=\{3,4,6\},
  \qquad
  \Delta_1=\{2,7\},
  \qquad
  \Delta_2=\{1\}
  \pmod 8.
  \label{eq:overview-direct-displacements}
\end{equation}
Masks remove wraparound terms and select the destinations belonging to each
class.  The direct topology therefore costs $3+2+1=6$ rotations.  For
$n=2^m$, the same borrow-chain structure produces $m-t$ displacement classes
in round $t$ and the triangular count in
\cref{eq:direct-triangular}.

This cost does not follow from the prefix function itself.  It follows from a
stronger routing requirement imposed by the direct topology: every
destination asks for the value at one exact logical predecessor.  The new
scan replaces that requirement with a block invariant under which many
physical copies are interchangeable.

\subsection{Replicate complete block aggregates}
\label{sec:overview-replication}

Let $\diamond$ denote ordered concatenation: $A\diamond X$ aggregates a lower
logical interval represented by $A$ followed by its adjacent upper interval
represented by $X$.  Only associativity is assumed; in particular,
$\diamond$ need not be commutative.

At level $d$, the scan maintains two packed states.
\begin{itemize}[leftmargin=*,itemsep=0.2em,topsep=0.25em]
  \item $A^{(d)}$ stores the complete aggregate of the aligned logical
  block of length $2^d$ containing the current position.  The same aggregate
  is replicated at every slot belonging to that block.
  \item $E^{(d)}$ stores the exclusive ordered prefix from the beginning of
  that block to, but not including, the current logical position.
\end{itemize}
Initially, $A^{(0)}$ contains the input states and $E^{(0)}$ contains the
monoid identity.  At level $d$, rotate only the replicated aggregate state:
\begin{equation}
  B^{(d)}=\Rot_{s_d}\bigl(A^{(d)}\bigr),
  \qquad
  s_d=2^{m-d-1}.
  \label{eq:overview-rotation}
\end{equation}
For a position in the lower child, $B^{(d)}$ supplies the upper-child
aggregate; for a position in the upper child, it supplies the lower-child
aggregate.  Public masks select the two cases.  Informally, the updates are
\begin{align}
  \text{lower child:}\quad
  &A^{(d+1)}=A^{(d)}\diamond B^{(d)},
  &E^{(d+1)}=E^{(d)},
  \label{eq:overview-lower-update}\\
  \text{upper child:}\quad
  &A^{(d+1)}=B^{(d)}\diamond A^{(d)},
  &E^{(d+1)}=B^{(d)}\diamond E^{(d)}.
  \label{eq:overview-upper-update}
\end{align}
The operand order always follows the logical order ``lower interval, then
upper interval.''  The same update therefore works for arbitrary
noncommutative monoids.

\begin{figure}[t]
  \centering
  \setlength{\tabcolsep}{3.3pt}
  \renewcommand{\arraystretch}{1.2}
  \small
  \begin{tabular}{@{}lcccccccc@{}}
    \toprule
    physical slot
      & $0$ & $1$ & $2$ & $3$ & $4$ & $5$ & $6$ & $7$ \\
    logical index
      & $0$ & $4$ & $2$ & $6$ & $1$ & $5$ & $3$ & $7$ \\
    \midrule
    $A^{(0)}$
      & $[0]$ & $[4]$ & $[2]$ & $[6]$ & $[1]$ & $[5]$ & $[3]$ & $[7]$ \\
    $A^{(1)}$ after $\Rot_{4}$
      & $[0{:}1]$ & $[4{:}5]$ & $[2{:}3]$ & $[6{:}7]$
      & $[0{:}1]$ & $[4{:}5]$ & $[2{:}3]$ & $[6{:}7]$ \\
    $A^{(2)}$ after $\Rot_{2}$
      & $[0{:}3]$ & $[4{:}7]$ & $[0{:}3]$ & $[4{:}7]$
      & $[0{:}3]$ & $[4{:}7]$ & $[0{:}3]$ & $[4{:}7]$ \\
    $A^{(3)}$ after $\Rot_{1}$
      & $[0{:}7]$ & $[0{:}7]$ & $[0{:}7]$ & $[0{:}7]$
      & $[0{:}7]$ & $[0{:}7]$ & $[0{:}7]$ & $[0{:}7]$ \\
    \bottomrule
  \end{tabular}
  \caption{Replication of complete aggregates in the eight-slot example.
  The notation $[a{:}b]$ denotes the ordered aggregate of logical inputs
  $a,\ldots,b$.  Each level uses one rotation.  The exclusive-prefix state,
  omitted from the table, is updated in parallel only for positions in the
  upper child.}
  \label{fig:eight-slot-replication}
\end{figure}

\Cref{fig:eight-slot-replication} shows the aggregate invariant for the
physical order in \cref{eq:overview-bitrev-eight}.  The rotations are
$4,2,1$.  After the first level, each logical pair aggregate has two physical
copies; after the second, each four-element aggregate has four copies; after
the third, the complete aggregate fills the domain.  In parallel, the
exclusive state accumulates the complete lower child exactly when the
current position belongs to the upper child.  After the final level,
$E_i^{(m)}$ is the ordered aggregate of all inputs preceding logical position
$i$.

The aggregate update is unnecessary after the final level.  The scan thus
uses one aggregate composition and one exclusive-prefix composition at each
of the first $m-1$ levels, followed by one final prefix composition:
\begin{equation}
  T\leq 2\cdot(m-1)+1=2\cdot m-1.
  \label{eq:overview-monoid-work}
\end{equation}
This is a count of packed monoid evaluations, not of ciphertext
multiplications.  A concrete carry encoding may require several encrypted
products inside one monoid evaluation.

\subsection{Why one rotation supplies both siblings}
\label{sec:overview-sibling}

The geometric reason for \cref{eq:overview-rotation} is especially simple
under bit reversal.  Write a logical index at level $d$ as
\begin{equation}
  i=u\cdot 2^{d+1}+b\cdot 2^d+r,
  \qquad b\in\{0,1\},\quad 0\leq r<2^d,
  \label{eq:overview-index-decomposition}
\end{equation}
and set $s_d=2^{m-d-1}$.  For a fixed parent block, its two children occupy
physical cosets of the form
\begin{align}
  \mathcal C_0
    &=\{c+2\cdot s_d\cdot k:0\leq k<2^d\},
  \label{eq:overview-coset-zero}\\
  \mathcal C_1
    &=\{c+s_d+2\cdot s_d\cdot k:0\leq k<2^d\}.
  \label{eq:overview-coset-one}
\end{align}
Translation by $s_d$ swaps these cosets modulo $2^m$.  It need not map a
position to the physical image of its exact XOR partner: wraparound may
select a different position in the sibling coset.  Exact partners are
unnecessary because $A^{(d)}$ is constant over each child.  Any source in the
sibling coset carries the same complete aggregate.

This is the decisive distinction between the two scans.  The direct topology
routes syntactically designated predecessors.  The replicated topology routes
semantic block values.  Public replication turns one directed cyclic
translation into the effect of a bidirectional hypercube exchange.

\subsection{Optimality and the structure of equality}
\label{sec:overview-optimality}

The lower-bound intuition is independent of the construction.  Fix an output
slot and suppose the circuit invokes rotations with offsets
$\delta_1,\ldots,\delta_R$.  A dependency path either traverses or avoids each
invocation, so its source displacement is a subset sum
\begin{equation}
  \sum_{j\in J}\delta_j\pmod{2^m}
  \qquad\text{for some }J\subseteq\{1,\ldots,R\}.
  \label{eq:overview-subset-sum}
\end{equation}
At most $2^R$ source displacements can reach the fixed output.  The complete
prefix depends on all $2^m$ inputs, which forces $R\geq m$.  Binary monoid
fan-in independently gives depth $D\geq m$.  The replicated scan meets both
bounds simultaneously.

Equality leaves no slack.  When $R=m$, every subset in
\cref{eq:overview-subset-sum} must produce a different residue modulo $2^m$.
The offsets consequently form a complete $2$-adic basis: after reordering,
there is exactly one offset of each valuation
\begin{equation}
  0,1,\ldots,m-1.
  \label{eq:overview-valuations}
\end{equation}
The formal proof evaluates the associated group-ring factorization at
$2$-power roots of unity.  This characterizes the rotation profile of every
optimal circuit, but not its entire topology; distinct circuits may realize
the same valuation basis.

\subsection{Trading direct keys for online rotations}
\label{sec:overview-key-frontier}

The $m$ optimal calls may use $m$ directly supported offsets.  Reducing this
key set forces some scale rotations to be synthesized by repeated calls to a
smaller set.  If key $g_j$ is invoked $c_j$ times, a dependency path may use
it between zero and $c_j$ times.  The resulting coefficient box contains at
most
\begin{equation}
  \prod_{j=1}^{K}(c_j+1)
  \label{eq:overview-key-box}
\end{equation}
possible displacements.  Coverage of all $2^m$ inputs therefore requires the
product in \cref{eq:overview-key-box} to be at least $2^m$.

When $K$ divides $m$, the lower bound is attained by dividing the valuation
range into $K$ equal consecutive blocks and directly keying the first scale
of each block.  For example, the exact calls for $m=8$ are shown in
\cref{tab:overview-key-example}.
\begin{table}[t]
  \centering
  \caption{Exact online-call frontier for $m=8$ and direct-key budgets
  dividing $m$.  Fewer directly supported rotations reduce key material but
  require sequential synthesis of the omitted scales.}
  \label{tab:overview-key-example}
  \setlength{\tabcolsep}{10pt}
  \renewcommand{\arraystretch}{1.25}
  \begin{tabular}{@{}lcccc@{}}
    \toprule
    direct keys $K$ & $8$ & $4$ & $2$ & $1$ \\
    \midrule
    minimum calls $R^\star(8,K)$ & $8$ & $12$ & $30$ & $255$ \\
    \bottomrule
  \end{tabular}
\end{table}
For nondivisible $K$, the paper gives a sharp coefficient-counting lower
bound and a balanced dyadic construction.  We state exactness only when the
bounds coincide; the remaining gaps are not needed for the scan or carry
results.

\subsection{Exclusive prefixes for carry and borrow}
\label{sec:overview-carry}

A radix digit induces a transition $F_{S_i}:\{0,1\}\rightarrow\{0,1\}$ on
its incoming carry.  Ordered composition is noncommutative because the
lower-digit transition must act before the upper-digit transition.  The
exclusive scan returns the transition
\begin{equation}
  E_i=S_0\diamond\cdots\diamond S_{i-1}.
  \label{eq:overview-exclusive-state}
\end{equation}
The carry entering digit $i$ is $F_{E_i}(0)$.  Applying the local transition
$S_i$ then gives the outgoing carry, so both values required for digit
correction are available without communication.  An inclusive scan would
give the outgoing carry directly but would leave the incoming carry at the
preceding logical digit, reintroducing a layout-sensitive shift.

The scan theorem concerns the prefix-and-correction layer.  A concrete
\CKKS realization must still choose a carry-state encoding and account for
ciphertext products, conjugations, rescaling, key-switch error, and level
alignment.  The practical comparison is therefore not ``$m$ rotations versus
$m\cdot(m+1)/2$ rotations'' in isolation.  It is the exchange
\begin{equation}
  \text{fewer rotations}
  \quad\longleftrightarrow\quad
  \text{two live states and additional monoid work}.
  \label{eq:overview-resource-exchange}
\end{equation}
The evaluation will determine where this exchange improves latency, memory,
and precision in a concrete backend.

\section{Model and Preliminaries}
\label{sec:model}

This section fixes the ordered-prefix problem, the packed circuit model, and
our layout conventions.  The abstraction deliberately separates three
resources that are often conflated in homomorphic implementations: binary
monoid depth, packed monoid evaluations, and cyclic-rotation invocations.
Concrete ciphertext multiplications, rescaling, and key switching refine
these abstract costs in the later \CKKS realization.

\subsection{Ordered monoid prefixes}
\label{sec:model-monoid}

Let $\mathcal M=(M,\diamond,e)$ be a monoid.  We assume associativity but not
commutativity.  The notation $A\diamond X$ always means that the logical
interval represented by $A$ precedes the adjacent interval represented by
$X$.  This operand convention remains fixed even when a concrete state
representation writes function composition in the opposite syntactic order.

For an input sequence $x_0,\ldots,x_{n-1}\in M$, the inclusive and exclusive
ordered prefixes are
\begin{align}
  \Pi_i
    &\defeq x_0\diamond x_1\diamond\cdots\diamond x_i,
  \label{eq:model-inclusive-prefix}\\
  \Pi_i^{\circ}
    &\defeq
      \begin{cases}
        e, & i=0,\\
        x_0\diamond x_1\diamond\cdots\diamond x_{i-1}, & i>0.
      \end{cases}
  \label{eq:model-exclusive-prefix}
\end{align}
Parenthesization is immaterial by associativity, whereas the left-to-right
order is part of the required output.  Unless stated otherwise, the active
length is
\begin{equation}
  n=2^m
  \qquad\text{for an integer }m\geq 2.
  \label{eq:model-active-length}
\end{equation}
For $0\leq d\leq m$ and $0\leq u<2^{m-d}$, define the aligned logical block
\begin{equation}
  \mathcal B_{u,d}
  \defeq
  \{u\cdot 2^d,\ldots,(u+1)\cdot 2^d-1\}.
  \label{eq:model-aligned-block}
\end{equation}
We write $[a{:}b]$ for the ordered aggregate
$x_a\diamond\cdots\diamond x_b$ and take $[a{:}a-1]=e$.  This interval
notation is semantic: a circuit value may be stored in any physical slot and
may have several identical copies.

\subsection{Packed cyclic-rotation circuits}
\label{sec:model-circuits}

Physical slots are indexed by the cyclic group $\Z_n$.  A packed monoid state
is a vector $V\in M^n$.  The circuit is an acyclic, data-independent directed
graph built from the following operations.

\begin{definition}[Packed prefix circuit]
\label{def:packed-prefix-circuit}
A packed prefix circuit may use:
\begin{enumerate}[leftmargin=*,label=(\roman*),itemsep=0.2em,topsep=0.25em]
  \item input vectors and public constant vectors, including the all-identity
  vector;
  \item free fanout, copying, and public slotwise selection between already
  available states;
  \item a packed monoid gate
  \begin{equation}
    \operatorname{Comp}(U,V)[p]
      \defeq U[p]\diamond V[p]
      \qquad\text{for every }p\in\Z_n;
    \label{eq:model-packed-composition}
  \end{equation}
  \item a cyclic rotation by offset $\delta\in\Z_n$,
  \begin{equation}
    \rho_{\delta}(V)[p]
      \defeq V[p-\delta\bmod n].
    \label{eq:model-rotation-convention}
  \end{equation}
\end{enumerate}
The output must contain the requested logical prefixes in the prescribed
physical layout.
\end{definition}

Public selection captures the effect of plaintext masks without assigning a
monoid meaning to multiplication by zero.  It does not combine two
independent monoid words.  The only gate that merges monoid-bearing values is
the binary gate in \cref{eq:model-packed-composition}.  Representation-level
arithmetic used to implement one such gate is accounted for separately.

We use the following cost measures.
\begin{itemize}[leftmargin=*,itemsep=0.25em,topsep=0.25em]
  \item The \emph{monoid depth} $D$ is the maximum number of packed monoid
  gates on a directed input--output path.
  \item The \emph{packed monoid work} $T$ is the total number of invocations
  of $\operatorname{Comp}$.
  \item The \emph{rotation count} $R$ is the total number of rotation-gate
  invocations.  Rotating one state once costs one call, independently of the
  number of slots whose values are subsequently selected.  A rotated result
  may be fanned out and reused without another rotation call.
  \item The \emph{direct-key count} $K$ is the number of distinct nonzero
  offsets for which the evaluator holds a directly usable rotation key.  If
  the directly keyed offsets are $g_1,\ldots,g_K$ and offset $g_j$ is invoked
  $c_j$ times, then
  \begin{equation}
    R=\sum_{j=1}^{K}c_j.
    \label{eq:model-rotation-count-by-key}
  \end{equation}
  An effective rotation synthesized by several directly keyed calls is
  charged once for every call in that sequence.
\end{itemize}

This model counts online cyclic translations, not wall-clock latency.  It
does not identify distinct rotations that share a hoisted decomposition, and
it does not charge the internal ciphertext products of a monoid gate.  It
also excludes hidden data movement: a dense linear transform, arbitrary slot
permutation, coefficient-to-slot transform, or bootstrapping transform may
be used by a larger application, but it is not a free operation inside the
prefix circuit.  Any automorphism that changes slot positions must be
represented by counted routing calls.

\subsection{Bit reversal and aligned physical cosets}
\label{sec:model-layout}

Bit reversal does not keep an aligned logical block contiguous in physical
memory.  Instead, it turns the block into a regular strided set.  The reason
is easiest to see at the bit level.  Inside a logical block of length $2^d$,
the lower $d$ bits of the logical index vary, whereas the upper $m-d$ bits
are fixed.  Bit reversal exchanges these roles: the varying lower logical
bits become the upper $d$ physical bits, and the fixed upper logical bits
become the lower $m-d$ physical bits.  Consequently, all physical slots of
the block have the same residue modulo $2^{m-d}$ and differ only by multiples
of $2^{m-d}$.  This regularity is the geometric property used by the scan.

To make the observation precise, write the $m$-bit expansion of
$i\in\{0,\ldots,2^m-1\}$ as
\begin{equation}
  i=\sum_{j=0}^{m-1} i_j\cdot 2^j,
  \qquad i_j\in\{0,1\}.
  \label{eq:model-binary-expansion}
\end{equation}
The $m$-bit reversal permutation is
\begin{equation}
  \rev_m(i)
  \defeq
  \sum_{j=0}^{m-1} i_j\cdot 2^{m-1-j},
  \label{eq:model-bit-reversal}
\end{equation}
and logical position $i$ is stored at physical slot
\begin{equation}
  L(i)=\rev_m(i).
  \label{eq:model-layout-map}
\end{equation}

For each $d$, define
\begin{equation}
  H_d
  \defeq
  \{2^{m-d}\cdot k:0\leq k<2^d\}
  \subseteq \Z_{2^m}.
  \label{eq:model-hd}
\end{equation}
The set $H_d$ is the subgroup of physical addresses whose lower $m-d$ bits
are zero.  A coset $c+H_d$ therefore fixes those lower bits to the residue
$c$ and lets the upper $d$ bits vary.  The preceding bit-level observation
now becomes the exact identity
\begin{equation}
  L(\mathcal B_{u,d})
  =\rev_{m-d}(u)+H_d.
  \label{eq:model-block-coset}
\end{equation}
Indeed, if $i=u\cdot 2^d+r$ with $0\leq r<2^d$, then
\begin{equation}
  \rev_m(u\cdot 2^d+r)
  =\rev_{m-d}(u)+2^{m-d}\cdot\rev_d(r).
  \label{eq:model-block-reversal-identity}
\end{equation}
Thus an aligned logical block is not physically contiguous, but all of its
slots lie in one explicitly known coset.

The same bit view explains the sibling displacement.  Consider a parent
block of length $2^{d+1}$.  Its two children differ only in logical bit $d$.
After reversal, this bit becomes physical bit $m-d-1$, so changing the child
orientation changes the physical address by
\begin{equation}
  s_d=2^{m-d-1}.
  \label{eq:model-sibling-offset}
\end{equation}
The lower $d$ logical bits still vary inside each child and become steps of
size $2^{m-d}=2\cdot s_d$.  Hence, for a parent-dependent residue $c$, the
two child supports have the form
\begin{align}
  \mathcal C_0
    &=\{c+2\cdot s_d\cdot k:0\leq k<2^d\},
  \label{eq:model-child-coset-zero}\\
  \mathcal C_1
    &=\{c+s_d+2\cdot s_d\cdot k:0\leq k<2^d\}.
  \label{eq:model-child-coset-one}
\end{align}
Under the convention in \cref{eq:model-rotation-convention}, subtracting
$s_d$ from an address in either child yields an address in the other child,
modulo $2^m$.  Therefore $\rho_{s_d}$ exchanges the two physical supports.
It may permute positions within the sibling support rather than select the
exact XOR partner.  This distinction is harmless once the complete child
aggregate is replicated across that support: every source copy represents
the same monoid element.

\paragraph{Eight-slot example.}
For $m=3$, the physical order is $0,4,2,6,1,5,3,7$.  At the first level,
logical pairs occupy supports such as $\{0,4\}$ and are exchanged with offset
$4$.  At the next level, the two pair supports inside logical block
$[0{:}3]$ are $\{0,4\}$ and $\{2,6\}$, separated by offset $2$.  At the last
level, logical halves $[0{:}3]$ and $[4{:}7]$ occupy the even and odd physical
slots, separated by offset $1$.  The scan rotations are therefore
$4,2,1$: each level changes exactly the physical bit that records the current
logical child orientation.

For public masks, define the level-$d$ orientation bit by
\begin{equation}
  \mu_d[L(i)]\defeq i_d.
  \label{eq:model-orientation-mask}
\end{equation}
Thus $\mu_d=0$ selects the logically lower child and $\mu_d=1$ selects the
logically upper child, even though their slots are interleaved in physical
order.

\paragraph{Fully occupied digit-major interleaving.}
The same geometry executes $g$ independent scans in one cyclic domain.  Let
$0\leq r<g$ identify an object and let $0\leq i<n$ identify a digit.  The
layout
\begin{equation}
  \Phi(r,i)
  \defeq
  r+g\cdot\rev_m(i)
  \pmod{g\cdot n}
  \label{eq:model-segmented-layout}
\end{equation}
fills all $g\cdot n$ slots.  A level-$d$ sibling exchange uses offset
\begin{equation}
  g\cdot s_d=g\cdot 2^{m-d-1}.
  \label{eq:model-segmented-offset}
\end{equation}
This offset preserves the residue modulo $g$, so every object executes the
same scan level independently.  The statement requires a complete cyclic
domain, or equivalent identity padding.  It does not automatically cover
partially filled or object-major contiguous layouts, where wraparound can
cross object boundaries.

\subsection{A representation boundary for coordinatewise products}
\label{sec:model-representation}

The prefix theorem concerns public layouts, not arbitrary encrypted linear
representations.  The following elementary proposition records why
permutations and public coordinate scalings form the natural boundary for a
generic coordinatewise product.

Let $\mathbb F$ be a field and define
\begin{equation}
  \operatorname{cmul}(\boldsymbol{x},\boldsymbol{y})_j
  \defeq x_j\cdot y_j.
  \label{eq:model-coordinate-product}
\end{equation}

\begin{proposition}[Monomial representation boundary]
\label{prop:monomial-boundary}
Let $E\in\operatorname{GL}_n(\mathbb F)$.  Suppose there exists an invertible
diagonal matrix $\Lambda$ such that, for all
$\boldsymbol{x},\boldsymbol{y}\in\mathbb F^n$,
\begin{equation}
  E\cdot\operatorname{cmul}(\boldsymbol{x},\boldsymbol{y})
  =\Lambda\cdot
   \operatorname{cmul}(E\cdot\boldsymbol{x},E\cdot\boldsymbol{y}).
  \label{eq:model-native-product}
\end{equation}
Then $E$ is monomial: $E=D\cdot P$ for an invertible diagonal matrix $D$ and
a permutation matrix $P$.  Conversely, every such $E$ satisfies
\cref{eq:model-native-product} with $\Lambda=D^{-1}$.
\end{proposition}

\begin{proof}
Fix a row $a=(a_0,\ldots,a_{n-1})$ of $E$, and let $\lambda\neq 0$ be the
corresponding diagonal entry of $\Lambda$.  Comparing the coefficient of
$x_j\cdot y_k$ for $j\neq k$ in \cref{eq:model-native-product} gives
$\lambda\cdot a_j\cdot a_k=0$.  Hence every row of $E$ contains at most one
nonzero entry.  Invertibility forces exactly one nonzero entry in every row
and every column, so $E=D\cdot P$.  For the converse,
\begin{align*}
  D^{-1}\cdot
  \operatorname{cmul}(D\cdot P\cdot\boldsymbol{x},
                      D\cdot P\cdot\boldsymbol{y})
  &=D\cdot P\cdot
    \operatorname{cmul}(\boldsymbol{x},\boldsymbol{y})\\
  &=E\cdot\operatorname{cmul}(\boldsymbol{x},\boldsymbol{y}),
\end{align*}
which proves the claim.
\end{proof}

The proposition does not rule out useful mixing transforms.  It says that an
arbitrary invertible mixing transform cannot preserve a generic
coordinatewise product through only coordinatewise arithmetic and public
scale correction.  This paper therefore analyzes layout permutations
explicitly and treats additional linear transforms as separately counted
operations.

\subsection{Carry transitions as an ordered monoid}
\label{sec:model-carry}

We recall the transition semantics used by radix carry constructions~\cite{ChaPL26}.  Fix a
radix $B\geq 2$.  Consider a provisional digit
\begin{equation}
  z_i\in\{0,\ldots,2\cdot B-2\}
  \label{eq:model-provisional-digit-range}
\end{equation}
and an incoming carry $c\in\{0,1\}$.  The digit induces the transition
\begin{equation}
  F_{z_i}(c)
  \defeq
  \left\lfloor\frac{z_i+c}{B}\right\rfloor.
  \label{eq:model-carry-transition}
\end{equation}
It is one of the three monotone Boolean maps
\begin{align}
  \mathsf K(c)&=0
    &&\text{if }0\leq z_i\leq B-2,\label{eq:model-kill}\\
  \mathsf P(c)&=c
    &&\text{if }z_i=B-1,\label{eq:model-propagate}\\
  \mathsf G(c)&=1
    &&\text{if }B\leq z_i\leq 2\cdot B-2.\label{eq:model-generate}
\end{align}
For a lower interval state $A$ followed by an adjacent upper interval state
$X$, ordered concatenation represents function composition in the order
\begin{equation}
  F_{A\diamond X}=F_X\circ F_A.
  \label{eq:model-transition-order}
\end{equation}
Associativity of $\diamond$ follows immediately from associativity of
function composition.  The identity is the propagate map
$e(c)=c$.

For algebraic reasoning, a transition may be represented by a homogeneous
affine state
\begin{equation}
  S=(p,\gamma,h),
  \qquad
  F_S(c)=\frac{\gamma+p\cdot c}{h}.
  \label{eq:model-projective-state}
\end{equation}
The three Boolean transitions have representatives
\begin{equation}
  \mathsf K=(0,0,1),
  \qquad
  \mathsf P=(1,0,1),
  \qquad
  \mathsf G=(0,1,1).
  \label{eq:model-projective-kpg}
\end{equation}
If $A=(p_A,\gamma_A,h_A)$ is lower and
$X=(p_X,\gamma_X,h_X)$ is upper, then
\begin{equation}
  A\diamond X
  =\bigl(
      p_X\cdot p_A,
      \gamma_X\cdot h_A+p_X\cdot\gamma_A,
      h_X\cdot h_A
    \bigr).
  \label{eq:model-projective-composition}
\end{equation}
Equation~\eqref{eq:model-projective-composition} is associative because it
represents \cref{eq:model-transition-order}.  Multiplying all three
coordinates by a common nonzero scalar does not change the represented
transition.

Let $S_i$ be the state of digit $i$, and let the exclusive state be
\begin{equation}
  E_i
  \defeq
  \begin{cases}
    e, & i=0,\\
    S_0\diamond\cdots\diamond S_{i-1}, & i>0.
  \end{cases}
  \label{eq:model-exclusive-carry-state}
\end{equation}
Then the incoming carry, outgoing carry, and corrected digit are evaluated
slotwise as
\begin{equation}
  c_i=F_{E_i}(0),
  \qquad
  c_{i+1}=F_{S_i}(c_i),
  \qquad
  d_i=z_i+c_i-B\cdot c_{i+1}.
  \label{eq:model-carry-correction}
\end{equation}
The scan is responsible for producing $E_i$; the three evaluations in
\cref{eq:model-carry-correction} require no further cross-slot routing.  A
borrow transition is obtained analogously by mapping an incoming borrow to
the outgoing borrow.  The same ordered-prefix circuit therefore applies to
borrow propagation, comparison through the final borrow, and conditional
subtraction.  The later radix-\CKKS instantiation distinguishes this abstract
transition monoid from the ciphertext operations used by a concrete encoding.

\section{Lower Bounds and Rigidity}
\label{sec:lower-bounds}

This section establishes the optimality statements used throughout the paper.
The argument has three layers.  First, binary monoid gates limit how rapidly
one value can acquire new dependencies.  Second, after expanding a packed
circuit into slot-labelled nodes, every input--output dependency path acquires
a physical displacement equal to a subset sum of the invoked rotation
offsets.  Third, equality in this counting argument forces those offsets to
form a complete $2$-adic basis.  We then return to the direct
transported-predecessor topology and determine its larger triangular cost.

The global results use only the circuit model of
\cref{def:packed-prefix-circuit}.  They permit arbitrary helper vectors,
fanout, public masks, branching, recomputation, and reuse of rotated states.
They count online cyclic-rotation invocations, not wall-clock time or the
number of distinct rotation keys.

\subsection{Dependency growth and monoid depth}
\label{sec:lower-depth}

For a circuit node and a fixed physical slot, its \emph{dependency set} is the
set of input monoid elements whose values can affect that slot.  A public
constant has the empty dependency set.  Copying and rotation preserve its
cardinality, while public slotwise selection chooses one already available
value at the considered slot.  The only operation that merges two
monoid-bearing dependency sets is the binary packed monoid gate.

\begin{lemma}[Binary dependency growth]
\label{lem:binary-dependency-growth}
A value at monoid depth at most $d$ depends on at most $2^d$ input monoid
elements.
\end{lemma}

\begin{proof}
We prove the claim by induction over the circuit in topological order.  At
monoid depth zero, a slot contains either one input value or a public
constant, and therefore depends on at most one input, which equals $2^0$.
Copying, rotation, and public selection do not merge two dependency sets, so
they preserve the claimed bound.

Consider now a packed monoid gate with input depths $d_0$ and $d_1$.  At a
fixed slot, let the two input dependency sets be $S_0$ and $S_1$.  By the
induction hypothesis,
\[
  |S_0|\leq 2^{d_0}
  \qquad\text{and}\qquad
  |S_1|\leq 2^{d_1}.
\]
The output dependency set is contained in $S_0\cup S_1$, so
\begin{align*}
  |S_0\cup S_1|
  &\leq |S_0|+|S_1| \\
  &\leq 2^{d_0}+2^{d_1} \\
  &\leq 2\cdot 2^{\max\{d_0,d_1\}}
   =2^{1+\max\{d_0,d_1\}}.
\end{align*}
The output depth is $1+\max\{d_0,d_1\}$, which proves the induction step.
\end{proof}

\begin{theorem}[Depth lower bound]
\label{thm:depth-lower-bound}
Let $n=2^m$ with $m\geq 2$.  Every packed circuit computing all inclusive or
all exclusive ordered prefixes has monoid depth
\[
  D\geq m.
\]
The same bound holds for any circuit producing one value that depends on all
$n$ inputs.
\end{theorem}

\begin{proof}
For an inclusive scan, the output at logical position $n-1$ is
$x_0\diamond\cdots\diamond x_{n-1}$ and therefore depends on all
$n=2^m$ inputs.  By \cref{lem:binary-dependency-growth}, a depth-$D$ output
can depend on at most $2^D$ inputs.  Hence $2^D\geq 2^m$, which implies
$D\geq m$.

For an exclusive scan, the output at logical position $n-1$ is
$x_0\diamond\cdots\diamond x_{n-2}$ and depends on $2^m-1$ inputs.  If
$D\leq m-1$, then \cref{lem:binary-dependency-growth} would give at most
$2^{m-1}$ dependencies.  Since $m\geq 2$ implies
$2^{m-1}<2^m-1$, this is impossible.  Thus $D\geq m$ also for the exclusive
scan.  The final statement follows from the inclusive argument whenever the
specified output depends on all $n$ inputs.
\end{proof}

\subsection{A global rotation lower bound}
\label{sec:lower-rotation}

The next argument tracks physical displacement rather than monoid depth.  Let
the circuit contain $R$ rotation gates, listed in a topological order, with
offsets
\[
  \delta_1,\ldots,\delta_R\in\Z_n.
\]
Repeated offsets are listed repeatedly because the cost measure counts gate
invocations.

\begin{lemma}[Subset-displacement bound]
\label{lem:subset-displacement}
Fix an output slot $p\in\Z_n$.  If the input at slot $q\in\Z_n$ can affect
that output, then
\begin{equation}
  p-q
  \equiv
  \sum_{j\in J}\delta_j
  \pmod n
  \label{eq:lower-path-subset-sum}
\end{equation}
for some subset $J\subseteq\{1,\ldots,R\}$.  Consequently, at most $2^R$
distinct input slots can affect one fixed output slot.
\end{lemma}

\begin{proof}
Expand every packed circuit node into $n$ slot-labelled nodes.  A local gate
has edges only between nodes carrying the same slot label.  Under the
rotation convention of \cref{eq:model-rotation-convention}, the slot-$r$
input of a gate $\rho_{\delta_j}$ is connected to its slot-$(r+\delta_j)$
output.  Thus crossing that gate increases the slot label by $\delta_j$
modulo $n$.

If input slot $q$ affects output slot $p$, the expanded dependency graph
contains a directed path from $q$ to $p$.  Let $J$ be the set of rotation-gate
invocations crossed by this path.  The circuit is acyclic, so a path cannot
cross the same gate invocation twice.  All local edges preserve the slot
label, while each crossed rotation gate contributes its offset.  Telescoping
the successive slot labels along the path therefore gives
\[
  p\equiv q+\sum_{j\in J}\delta_j\pmod n,
\]
which is \cref{eq:lower-path-subset-sum}.

There are $2^R$ subsets of the $R$ gate invocations.  Different subsets may
produce the same residue, but no additional displacement can arise.
Therefore at most $2^R$ input-slot displacements, and hence at most $2^R$
input slots, can reach a fixed output slot.
\end{proof}

\begin{theorem}[Global rotation lower bound]
\label{thm:rotation-lower-bound}
Let $n=2^m$ with $m\geq 2$.  Every packed circuit computing all inclusive or
all exclusive ordered prefixes requires
\[
  R\geq m
\]
cyclic-rotation invocations.  The same bound holds whenever one output depends
on all $n$ inputs.
\end{theorem}

\begin{proof}
For an inclusive scan, the output at logical position $n-1$ depends on all
$2^m$ input slots.  By \cref{lem:subset-displacement}, one fixed output can
have at most $2^R$ input-slot dependencies, and hence
$2^R\geq 2^m$.  Taking binary logarithms gives $R\geq m$.

For an exclusive scan, the output at logical position $n-1$ depends on
$2^m-1$ distinct input slots.  If $R\leq m-1$, then
\[
  2^R\leq 2^{m-1}<2^m-1,
\]
contradicting \cref{lem:subset-displacement}.  Therefore $R\geq m$.  If a
specified output depends on all $n$ inputs, the first counting argument
applies directly.
\end{proof}

The proof does not assume an interval normal form, fixed predecessor shifts,
one live state, or any particular placement of intermediate aggregates.  It
also remains valid when one rotated state is reused by several later gates.
It would cease to apply only if additional uncounted operations were allowed
to move dependencies between slots.

\subsection{Equality rigidity and complete two-adic bases}
\label{sec:lower-rigidity}

The counting argument leaves no slack when an output depending on all slots
is obtained with exactly $m$ rotations.  For a nonzero residue
$\delta\in\Z_{2^m}$, let $\nuTwo(\delta)$ denote the $2$-adic valuation of its
least nonnegative integer representative.

\begin{lemma}[Complete two-adic basis]
\label{lem:complete-two-adic-basis}
Let $\delta_0,\ldots,\delta_{m-1}\in\Z_{2^m}$.  The subset-sum map
\begin{equation}
  \varphi:\{0,1\}^m\longrightarrow\Z_{2^m},
  \qquad
  \varphi(b_0,\ldots,b_{m-1})
  =\sum_{j=0}^{m-1}b_j\cdot\delta_j
  \label{eq:lower-subset-sum-map}
\end{equation}
is bijective if and only if, after reordering the offsets,
\begin{equation}
  \nuTwo(\delta_j)=j
  \qquad\text{for every }0\leq j<m.
  \label{eq:lower-complete-valuations}
\end{equation}
\end{lemma}

\begin{proof}
Choose the representatives of the $\delta_j$ in
$\{0,\ldots,2^m-1\}$.

\emph{Necessity.}
Assume that $\varphi$ is bijective.  Expanding the product
$\prod_j(1+X^{\delta_j})$ produces one monomial for each subset of the
offsets.  In the quotient by $X^{2^m}-1$, exponents are reduced modulo
$2^m$.  Bijectivity says that every residue occurs as a subset sum exactly
once, and therefore
\begin{equation}
  \prod_{j=0}^{m-1}\bigl(1+X^{\delta_j}\bigr)
  =1+X+\cdots+X^{2^m-1}
  \quad\text{in }\Z[X]/(X^{2^m}-1).
  \label{eq:lower-group-ring-factorization}
\end{equation}

Fix $k\in\{0,\ldots,m-1\}$ and let $\zeta$ be a primitive
$2^{k+1}$-st root of unity.  Since $2^{k+1}$ divides $2^m$, one has
$\zeta^{2^m}=1$; moreover $\zeta\neq 1$.  Evaluation at $\zeta$ is
well-defined on the quotient ring, and the geometric series on the
right-hand side satisfies
\[
  1+\zeta+\cdots+\zeta^{2^m-1}
  =\frac{\zeta^{2^m}-1}{\zeta-1}
  =0.
\]
Hence the product on the left vanishes in $\mathbb{C}$, so at least one
factor vanishes.  For some $j$,
\[
  1+\zeta^{\delta_j}=0,
  \qquad\text{that is,}\qquad
  \zeta^{\delta_j}=-1.
\]
In the cyclic group generated by $\zeta$, the unique element of order two is
$\zeta^{2^k}=-1$.  Consequently
\[
  \delta_j\equiv 2^k\pmod{2^{k+1}},
\]
which is equivalent to $\nuTwo(\delta_j)=k$.  Thus every valuation
$k\in\{0,\ldots,m-1\}$ occurs among the $m$ offsets.  There are exactly $m$
offsets, so each valuation occurs exactly once.  In particular, no offset is
zero.

\emph{Sufficiency.}
Conversely, reorder the offsets and choose integer representatives
\[
  \delta_j=2^j\cdot u_j,
  \qquad u_j\text{ odd}.
\]
Suppose two binary vectors have the same image under $\varphi$.  Subtracting
the corresponding subset sums yields
\begin{equation}
  \sum_{j=0}^{m-1}\varepsilon_j\cdot\delta_j
  \equiv 0\pmod{2^m},
  \qquad
  \varepsilon_j\in\{-1,0,1\}.
  \label{eq:lower-signed-subset-relation}
\end{equation}
If the vectors differ, let $t$ be the least index for which
$\varepsilon_t\neq 0$.  Every term in
\cref{eq:lower-signed-subset-relation} is divisible by $2^t$, so division by
$2^t$ gives a congruence modulo $2^{m-t}$.  Reducing that congruence modulo
$2$ leaves
\[
  \varepsilon_t\cdot u_t\equiv 0\pmod 2,
\]
because every term with index greater than $t$ contains an additional factor
of two.  But $\varepsilon_t\in\{-1,1\}$ and $u_t$ is odd, so the left-hand
side is odd, a contradiction.  The two binary vectors are therefore equal.
Thus $\varphi$ is injective.  Its domain and codomain both contain $2^m$
elements, so it is bijective.
\end{proof}

\begin{corollary}[Rigidity of rotation-optimal equality]
\label{cor:rotation-equality-rigidity}
Suppose a packed circuit uses exactly $R=m$ rotation invocations and produces
one output depending on all $2^m$ input slots.  If the invoked offsets are
$\delta_0,\ldots,\delta_{m-1}$, then, after reordering,
\[
  \nuTwo(\delta_j)=j
  \qquad\text{for }0\leq j<m.
\]
In particular, every offset is nonzero, the offsets are pairwise distinct,
and their subset sums cover every residue exactly once.
\end{corollary}

\begin{proof}
Fix the physical output slot $p$.  As the source slot $q$ ranges over
$\Z_{2^m}$, the displacement $p-q$ ranges over every residue of
$\Z_{2^m}$ exactly once.  By hypothesis, every source affects the output, so
\cref{lem:subset-displacement} associates each of these residues with a
subset of the $m$ invoked offsets.  Therefore the subset-sum map in
\cref{eq:lower-subset-sum-map} is surjective.  Its domain and codomain both
have cardinality $2^m$, so it is bijective.  The valuation statement follows
from \cref{lem:complete-two-adic-basis}.  Distinct valuations imply that the
offsets are nonzero and pairwise distinct, while bijectivity gives the final
subset-sum claim.
\end{proof}

\Cref{cor:rotation-equality-rigidity} characterizes the offset profile, not
the complete topology of an optimal circuit.  The canonical powers of two
are only one representative: every family
$\delta_j=2^j\cdot u_j$ with odd $u_j$ has the same subset-sum property.
Whether such a family is embedded into a correct ordered-prefix topology is a
separate constructive question.

\subsection{Exact cost of transported predecessors}
\label{sec:lower-direct}

The global lower bound is linear, but the natural recursive-doubling scan is
more expensive when each logical predecessor is transported directly in the
bit-reversed layout.  At stage $t\in\{0,\ldots,m-1\}$, that topology asks
logical position $i\geq 2^t$ to read the current state at $i-2^t$.  Under the
layout $L(i)=\rev_m(i)$, define the required physical displacement
\begin{equation}
  \Delta_t(i)
  \defeq
  L(i)-L(i-2^t)
  \pmod{2^m}.
  \label{eq:lower-direct-displacement}
\end{equation}
One masked rotation can serve every destination having the same displacement,
but destinations with different displacements require different rotated
copies of the stage-input vector.

\begin{proposition}[Cyclic diagonals of one predecessor stage]
\label{prop:direct-stage-diagonals}
For fixed $t$, the set
\[
  \{\Delta_t(i):2^t\leq i<2^m\}
\]
contains exactly $m-t$ distinct residues.  Their $2$-adic valuations are
$0,1,\ldots,m-t-1$.
\end{proposition}

\begin{proof}
Write $i_r\in\{0,1\}$ for bit $r$ of $i$, with bit zero least significant.
Because $i\geq 2^t$, at least one bit at position $t$ or above is one.  Let
\[
  q=q(i,t)
  \defeq
  \min\{r\geq t:i_r=1\}.
\]
The bits $i_t,\ldots,i_{q-1}$ are zero and $i_q=1$.  Subtracting $2^t$
therefore changes bit $q$ from one to zero, changes bits
$t,\ldots,q-1$ from zero to one, and leaves every other bit unchanged.

Bit reversal sends logical bit $r$ to physical bit $m-1-r$.  Hence the
physical difference $L(i)-L(i-2^t)$ receives a positive contribution
$2^{m-1-q}$ from bit $q$ and a negative contribution
$2^{m-1-\ell}$ from each bit $\ell\in\{t,\ldots,q-1\}$.  Thus
\begin{align}
  \Delta_t(i)
  &\equiv
    2^{m-1-q}
    -\sum_{\ell=t}^{q-1}2^{m-1-\ell}
    \pmod{2^m}
  \notag\\
  &\equiv
    2^{m-1-q}
    -\bigl(2^{m-t}-2^{m-q}\bigr)
    \pmod{2^m}
  \notag\\
  &\equiv
    3\cdot 2^{m-1-q}-2^{m-t}
    \pmod{2^m}
  \notag\\
  &\equiv
    2^{m-1-q}\cdot\bigl(3-2^{q-t+1}\bigr)
    \pmod{2^m}.
  \label{eq:lower-direct-displacement-formula}
\end{align}
The displayed geometric-sum identity also covers $q=t$, where the sum is
empty.  Since $q-t+1\geq 1$, the factor
$3-2^{q-t+1}$ is odd.  Moreover $m-1-q<m$, so reduction modulo $2^m$ does not
change the exact $2$-adic valuation.  Therefore
\[
  \nuTwo\bigl(\Delta_t(i)\bigr)=m-1-q.
\]

The formula depends on $i$ only through $q$.  Conversely, every
$q\in\{t,\ldots,m-1\}$ occurs, for example by choosing $i=2^q$.  The
corresponding valuations are
$m-t-1,m-t-2,\ldots,0$, which are pairwise distinct.  Hence there is exactly
one displacement class for each such $q$, and the total number of classes is
$m-t$.
\end{proof}

\begin{corollary}[Triangular cost of the direct topology]
\label{cor:direct-triangular-cost}
The direct transported-predecessor realization of recursive doubling uses
exactly
\begin{equation}
  R_{\mathrm{direct}}
  =\sum_{t=0}^{m-1}(m-t)
  =\frac{m\cdot(m+1)}{2}
  \label{eq:lower-direct-triangular-cost}
\end{equation}
rotation invocations.
\end{corollary}

\begin{proof}
Fix a stage $t$ and its common stage-input vector.  A rotation by offset
$\delta$ presents at destination slot $p$ the value from slot $p-\delta$.
Public masks may discard destinations, but they cannot change this source
displacement.  Therefore one rotated copy of the stage-input vector can serve
only the destinations belonging to one displacement class.  By
\cref{prop:direct-stage-diagonals}, the $m-t$ nonempty classes require at
least $m-t$ rotated copies.

This lower bound is attained within the direct topology: for each displacement
class $\delta$, rotate the stage-input vector by $\delta$, retain exactly the
destinations of that class with its public mask, and select the resulting
masked copies slotwise.  Thus stage $t$ uses exactly $m-t$ rotations.  Summing
over $t=0,
\ldots,m-1$ yields
\[
  \sum_{t=0}^{m-1}(m-t)
  =m+(m-1)+\cdots+1
  =\frac{m\cdot(m+1)}{2}.
\]
\end{proof}

This is an exact cost for the stated direct topology, not a global lower
bound.  The replicated scan developed in the next section changes the
invariant: it routes an interchangeable copy of a sibling aggregate instead
of the syntactically designated predecessor.

\subsection{Why changing the fixed shifts is insufficient}
\label{sec:lower-fixed-shift}

For completeness, we isolate a restricted rigidity statement explaining why
a better fixed-predecessor schedule cannot remove the triangular cost.
Consider a one-state scan in which stage $j$ uses one positive logical
predecessor distance $s_j$, uniformly across all active destinations, with
public masks suppressing out-of-range predecessors.  Logical indices do not
wrap around in this restricted model.

\begin{proposition}[Rigidity of minimum-depth fixed shifts]
\label{prop:fixed-shift-rigidity}
Let $n=2^m$.  If $m$ fixed-predecessor stages make the final position depend
on all $n$ inputs, then the multiset of logical shifts is
\[
  \{s_0,\ldots,s_{m-1}\}
  =\{1,2,4,\ldots,2^{m-1}\}.
\]
Consequently, every such bit-reversed direct schedule incurs the triangular
rotation count in \cref{eq:lower-direct-triangular-cost}.
\end{proposition}

\begin{proof}
Along a dependency path ending at the final logical position, stage $j$
either keeps the current logical position or follows the predecessor edge of
length $s_j$.  Since each stage is encountered at most once, the total
backward logical distance of the path is a subset sum of
$s_0,\ldots,s_{m-1}$.  To make the final position depend on every input, each
distance in $\{0,1,\ldots,2^m-1\}$ must occur.  There are exactly $2^m$
subsets and exactly $2^m$ required distances.  Hence all subset sums are
distinct and, as ordinary nonnegative integers, are precisely
$0,1,\ldots,2^m-1$.

Sort the shifts as $a_0\leq\cdots\leq a_{m-1}$.  Since the smallest positive
subset sum is one, $a_0=1$.  Assume inductively that
\[
  a_j=2^j
  \qquad\text{for }0\leq j<r.
\]
The subset sums of these first $r$ shifts are then exactly the interval
$\{0,1,\ldots,2^r-1\}$.  If $a_r>2^r$, no subset using the remaining shifts
can sum to $2^r$, because every remaining shift is at least $a_r$; this would
create a gap.  If $a_r<2^r$, then the singleton subset $\{a_r\}$ duplicates
the already existing subset sum $a_r$ formed from the first $r$ shifts;
this contradicts uniqueness.  Therefore $a_r=2^r$.  Induction gives the
claimed multiset.

The order of the stages does not affect the sum of their individual direct
routing costs.  A stage using logical shift $2^t$ has exactly $m-t$ physical
displacement classes by \cref{prop:direct-stage-diagonals}.  Summing these
costs over $t=0,\ldots,m-1$ gives the triangular count in
\cref{eq:lower-direct-triangular-cost}.
\end{proof}

The proposition is deliberately restricted to one-state fixed-predecessor
scans.  It does not constrain circuits that maintain replicated aggregates,
change layouts internally, or use more general helper states.  The distinction
between this restricted triangular cost and the global lower bound $R\geq m$
is precisely what motivates the construction of the next section.

\section{Rotation-Optimal Replicated Scan}
\label{sec:replicated-scan}

The lower bounds of \cref{sec:lower-bounds} show that an ordered scan on
$n=2^m$ inputs cannot use fewer than $m$ monoid levels or fewer than $m$
cyclic rotations.  This section gives a construction attaining both bounds
simultaneously.  Its central invariant differs from a transported-predecessor
scan: each aligned block stores its complete aggregate at \emph{every}
physical slot of the block.  A destination therefore does not need the copy
held by one designated logical partner.  Any copy from the sibling block is
semantically interchangeable, and the bit-reversed coset geometry supplies
such a copy with one rotation per level.

We present the exclusive form because it gives each radix digit its carry-in
without a final predecessor shift.  The inclusive variant follows from the
same recurrence with a different initialization.

\subsection{Replicated aggregate and local-prefix states}
\label{sec:replicated-invariant}

For $0\leq d\leq m$ and logical position $i$, let
\begin{equation}
  \beta_d(i)
  \defeq
  2^d\cdot\left\lfloor\frac{i}{2^d}\right\rfloor
  \label{eq:replicated-block-start}
\end{equation}
be the first index of the aligned length-$2^d$ block containing $i$.  At the
beginning of level $d$, the scan maintains two logical views of packed state
vectors:
\begin{align}
  A_i^{(d)}
    &\defeq
      [\beta_d(i){:}\beta_d(i)+2^d-1],
  \label{eq:replicated-aggregate-invariant}\\
  E_i^{(d)}
    &\defeq
      [\beta_d(i){:}i-1].
  \label{eq:replicated-exclusive-invariant}
\end{align}
The state $A_i^{(d)}$ is the aggregate of the complete aligned block, and the
same monoid element is replicated at every position of that block.  The
state $E_i^{(d)}$ is the exclusive prefix from the beginning of that block to
position $i$.  In particular, it equals the identity at the first position of
each block.

The corresponding physical vectors store $A_i^{(d)}$ and $E_i^{(d)}$ at slot
$L(i)=\rev_m(i)$.  To keep the notation readable, all recurrence equations
below are written in logical coordinates; every selection is implemented by
the public physical mask $\mu_d$ from \cref{eq:model-orientation-mask}.

For a bit-valued public mask $\mu$ and packed states $U,V$, write
\begin{equation}
  \operatorname{Sel}_{\mu}(U,V)[p]
  \defeq
  \begin{cases}
    U[p], & \mu[p]=0,\\
    V[p], & \mu[p]=1.
  \end{cases}
  \label{eq:replicated-public-selection}
\end{equation}
This operation performs no monoid composition; it only chooses between two
already available values at each slot.

\subsection{One rotation supplies every sibling aggregate}
\label{sec:replicated-sibling-copy}

At level $d$, the two length-$2^d$ children of a length-$2^{d+1}$ parent
occupy the physical cosets in
\cref{eq:model-child-coset-zero,eq:model-child-coset-one}.  Their separation
is
\begin{equation}
  s_d=2^{m-d-1}.
  \label{eq:replicated-stage-offset}
\end{equation}
The following lemma formalizes the semantic-copy property used by the scan.

\begin{lemma}[Sibling-copy lemma]
\label{lem:replicated-sibling-copy}
Fix $d\in\{0,\ldots,m-1\}$.  Suppose that a packed vector $A^{(d)}$ stores,
at every logical position $i$, the complete aggregate of the aligned
length-$2^d$ block containing $i$.  Let
\begin{equation}
  B^{(d)}\defeq\rho_{s_d}(A^{(d)}).
  \label{eq:replicated-rotated-aggregate}
\end{equation}
Then, at every logical position $i$, the value $B_i^{(d)}$ is the complete
aggregate of the sibling length-$2^d$ child in the same length-$2^{d+1}$
parent block.
\end{lemma}

\begin{proof}
Fix a parent block and let $c$ be the residue appearing in
\cref{eq:model-child-coset-zero,eq:model-child-coset-one}.  Its lower and
upper child supports are
\begin{align*}
  \mathcal C_0
    &=\{c+2\cdot s_d\cdot k:0\leq k<2^d\},\\
  \mathcal C_1
    &=\{c+s_d+2\cdot s_d\cdot k:0\leq k<2^d\}.
\end{align*}
Under the rotation convention
$\rho_{s_d}(A)[p]=A[p-s_d\bmod 2^m]$, it suffices to prove that subtracting
$s_d$ exchanges these supports.

Let $p=c+2\cdot s_d\cdot k\in\mathcal C_0$.  Then
\[
  p-s_d
  =c+s_d+2\cdot s_d\cdot(k-1).
\]
If $k>0$, this is visibly in $\mathcal C_1$.  If $k=0$, replace $k-1$ by
$2^d-1$: the two expressions differ by
\[
  2\cdot s_d\cdot 2^d=2^m,
\]
so they represent the same slot modulo $2^m$.  Hence every destination in
$\mathcal C_0$ reads a source in $\mathcal C_1$.

Conversely, for $p=c+s_d+2\cdot s_d\cdot k\in\mathcal C_1$,
\[
  p-s_d=c+2\cdot s_d\cdot k\in\mathcal C_0.
\]
Thus the rotation exchanges the two child supports.  It may permute the
positions within the sibling support, but the hypothesis states that
$A^{(d)}$ is constant on each child support and equals that child's complete
aggregate.  Therefore every destination receives the required sibling
aggregate.
\end{proof}

The lemma is stronger than an exact-partner statement in the way needed here:
the rotated source may be a different logical position in the sibling block,
but all such positions carry the same aggregate.  This is what allows one
directed cyclic rotation to realize the two directions of the logical
sibling exchange simultaneously.

\subsection{Construction}
\label{sec:replicated-construction}

At level $d$, first rotate the replicated aggregate vector as in
\cref{eq:replicated-rotated-aggregate}.  For a lower-child position
($i_d=0$), the current aggregate precedes the rotated sibling aggregate.  For
an upper-child position ($i_d=1$), the rotated lower-child aggregate precedes
the current aggregate.  Public selection arranges these operands before one
packed monoid gate:
\begin{align}
  U_A^{(d)}
    &\defeq
      \operatorname{Sel}_{\mu_d}(A^{(d)},B^{(d)}),
  \label{eq:replicated-aggregate-left-operand}\\
  V_A^{(d)}
    &\defeq
      \operatorname{Sel}_{\mu_d}(B^{(d)},A^{(d)}),
  \label{eq:replicated-aggregate-right-operand}\\
  \widetilde A^{(d+1)}
    &\defeq
      \operatorname{Comp}(U_A^{(d)},V_A^{(d)}).
  \label{eq:replicated-aggregate-update}
\end{align}
For the exclusive prefix, lower-child positions keep their existing value,
whereas upper-child positions prepend the complete lower-child aggregate:
\begin{align}
  C_E^{(d)}
    &\defeq
      \operatorname{Comp}(B^{(d)},E^{(d)}),
  \label{eq:replicated-prefix-candidate}\\
  E^{(d+1)}
    &\defeq
      \operatorname{Sel}_{\mu_d}(E^{(d)},C_E^{(d)}).
  \label{eq:replicated-prefix-update}
\end{align}
The aggregate update and prefix candidate use only level-$d$ values and can
therefore be evaluated in parallel.  At the last level, the complete
length-$2^m$ aggregate is not needed after the prefixes are produced, so
\cref{eq:replicated-aggregate-update} is omitted.

\begin{figure}[t]
\centering
\fbox{\begin{minipage}{0.93\linewidth}
\small
\textbf{Input:} a packed vector $X$ with $X[L(i)]=x_i$; the public masks
$\mu_0,\ldots,\mu_{m-1}$; and the all-identity vector $\mathbf e$.

\textbf{Initialization:} $A\leftarrow X$ and $E\leftarrow\mathbf e$.

\begin{enumerate}[leftmargin=1.7em,label=\arabic*.,itemsep=0.25em,topsep=0.35em]
  \item For $d=0,\ldots,m-1$, set
  $s_d\leftarrow2^{m-d-1}$ and $B\leftarrow\rho_{s_d}(A)$.
  \item Compute $C_E\leftarrow\operatorname{Comp}(B,E)$ and set
  $E\leftarrow\operatorname{Sel}_{\mu_d}(E,C_E)$.
  \item If $d<m-1$, set
  \[
    U_A\leftarrow\operatorname{Sel}_{\mu_d}(A,B),
    \qquad
    V_A\leftarrow\operatorname{Sel}_{\mu_d}(B,A),
  \]
  and update $A\leftarrow\operatorname{Comp}(U_A,V_A)$.
\end{enumerate}

\textbf{Output:} $E$, where $E[L(i)]=\Pi_i^{\circ}$ for every logical
position $i$.
\end{minipage}}
\caption{Rotation-optimal exclusive scan in bit-reversed layout.  The prefix
and aggregate compositions at a nonfinal level depend only on the previous
level and may be evaluated in parallel.}
\label{fig:replicated-exclusive-scan}
\end{figure}

The order of Steps 2 and 3 in \cref{fig:replicated-exclusive-scan} is only
presentational: both must read the old level-$d$ states.  An implementation
must not overwrite $A$ or $E$ before all level-$d$ operands have been formed.

\subsection{Correctness}
\label{sec:replicated-correctness}

\begin{theorem}[Replicated exclusive-prefix invariant]
\label{thm:replicated-exclusive-correctness}
For every $d\in\{0,\ldots,m\}$ and logical position $i$, the conceptual
recurrence in
\cref{eq:replicated-aggregate-update,eq:replicated-prefix-update} satisfies
\begin{align}
  A_i^{(d)}
    &=[\beta_d(i){:}\beta_d(i)+2^d-1],
  \label{eq:replicated-correctness-aggregate}\\
  E_i^{(d)}
    &=[\beta_d(i){:}i-1].
  \label{eq:replicated-correctness-prefix}
\end{align}
In particular, after level $m$,
\begin{equation}
  E_i^{(m)}=[0{:}i-1]=\Pi_i^{\circ}
  \qquad\text{for every }0\leq i<n.
  \label{eq:replicated-final-exclusive-prefix}
\end{equation}
The aggregate equation is required only through $d=m-1$ by the optimized
algorithm, because the final aggregate update is omitted.
\end{theorem}

\begin{proof}
We prove the two equations simultaneously by induction on $d$.

For $d=0$, the aligned block containing $i$ is the singleton $\{i\}$.  The
initialization gives
\[
  A_i^{(0)}=x_i=[i{:}i]
\]
and
\[
  E_i^{(0)}=e=[i{:}i-1].
\]
Since $\beta_0(i)=i$, both invariants hold.

Assume they hold at some level $d<m$.  Fix a parent block of length
$2^{d+1}$ and write its first index as $b$.  Its lower child is
$[b,b+2^d-1]$ and its upper child is
$[b+2^d,b+2^{d+1}-1]$.  By the induction hypothesis, $A^{(d)}$ stores the
complete aggregate of each child at every position of that child.  By
\cref{lem:replicated-sibling-copy}, $B^{(d)}$ stores the complete sibling
aggregate at every position.

First suppose that $i$ lies in the lower child, so $i_d=0$.  The induction
hypothesis and the sibling-copy lemma give
\begin{align*}
  A_i^{(d)}&=[b{:}b+2^d-1],\\
  B_i^{(d)}&=[b+2^d{:}b+2^{d+1}-1].
\end{align*}
The public operand selections in
\cref{eq:replicated-aggregate-left-operand,eq:replicated-aggregate-right-operand}
therefore place the lower-child aggregate first and the upper-child aggregate
second.  Associativity and the interval convention imply
\[
  \widetilde A_i^{(d+1)}
  =[b{:}b+2^d-1]\diamond[b+2^d{:}b+2^{d+1}-1]
  =[b{:}b+2^{d+1}-1].
\]
For the prefix update, $\mu_d[L(i)]=0$, so
\[
  E_i^{(d+1)}=E_i^{(d)}=[b{:}i-1].
\]
Because the parent and lower child have the same first index $b$, these are
exactly the two level-$(d+1)$ invariants.

Now suppose that $i$ lies in the upper child, so $i_d=1$.  Then
\begin{align*}
  B_i^{(d)}&=[b{:}b+2^d-1],\\
  A_i^{(d)}&=[b+2^d{:}b+2^{d+1}-1],\\
  E_i^{(d)}&=[b+2^d{:}i-1].
\end{align*}
The selected aggregate operands again appear in logical order, and hence
\[
  \widetilde A_i^{(d+1)}
  =[b{:}b+2^d-1]\diamond[b+2^d{:}b+2^{d+1}-1]
  =[b{:}b+2^{d+1}-1].
\]
Since $\mu_d[L(i)]=1$, the prefix update selects
$C_{E,i}^{(d)}=B_i^{(d)}\diamond E_i^{(d)}$.  Therefore
\[
  E_i^{(d+1)}
  =[b{:}b+2^d-1]\diamond[b+2^d{:}i-1]
  =[b{:}i-1].
\]
This is the exclusive prefix from the beginning of the parent block to $i$.
Thus both invariants hold at level $d+1$ for both child orientations.

At $d=m$, the unique aligned block begins at
$\beta_m(i)=0$.  Substituting this value into
\cref{eq:replicated-correctness-prefix} gives
$E_i^{(m)}=[0{:}i-1]=\Pi_i^{\circ}$, as claimed.
\end{proof}

\paragraph{Inclusive variant.}
Initialize a prefix state $P^{(0)}\leftarrow X$ instead of
$E^{(0)}\leftarrow\mathbf e$, and apply the same update
\[
  P^{(d+1)}
  =\operatorname{Sel}_{\mu_d}
    \bigl(P^{(d)},\operatorname{Comp}(B^{(d)},P^{(d)})\bigr).
\]
The identical induction, with local interval $[\beta_d(i){:}i]$ in place of
$[\beta_d(i){:}i-1]$, gives $P_i^{(m)}=\Pi_i$.  The operation counts are
unchanged.

\subsection{Complexity and exact optimality}
\label{sec:replicated-complexity}

\begin{theorem}[Cost of the replicated scan]
\label{thm:replicated-cost}
For $n=2^m$ with $m\geq2$, the exclusive or inclusive replicated scan uses
\begin{equation}
  D=m,
  \qquad
  R=m,
  \qquad
  T=2\cdot m-1.
  \label{eq:replicated-cost}
\end{equation}
It maintains two persistent monoid-state vectors, together with transient
rotated and composition results.
\end{theorem}

\begin{proof}
There is exactly one invocation of $\rho_{s_d}$ at each of the $m$ levels,
so $R=m$.

At every nonfinal level $d\in\{0,\ldots,m-2\}$, one packed monoid gate forms
$\widetilde A^{(d+1)}$ and one forms $C_E^{(d)}$.  These two gates use only
level-$d$ inputs and may be evaluated in parallel.  At the final level, only
$C_E^{(m-1)}$ is needed.  Hence
\[
  T=2\cdot(m-1)+1=2\cdot m-1.
\]

Along every input--output path, at most one monoid gate is crossed per level,
because the aggregate and prefix gates of the same level are parallel rather
than sequential.  Thus $D\leq m$.  For the final logical position $i=n-1$,
every bit $i_d$ equals one, so its prefix state is updated by a monoid gate at
every level.  The resulting path has depth $m$, and therefore $D=m$.

The algorithm carries the current aggregate vector $A$ and prefix vector $E$
(or $P$) from one level to the next.  The rotated sibling vector and the two
composition outputs need only be transient, although an implementation may
retain additional buffers to avoid destructive updates.
\end{proof}

\begin{corollary}[Simultaneous depth and rotation optimality]
\label{cor:replicated-exact-optimality}
In the packed circuit model of \cref{def:packed-prefix-circuit}, for
$n=2^m$ and $m\geq2$, the minimum monoid depth and minimum rotation count for
all ordered inclusive or exclusive prefixes satisfy
\begin{equation}
  D^{\star}(m)=R^{\star}(m)=m.
  \label{eq:replicated-exact-optimum}
\end{equation}
\end{corollary}

\begin{proof}
The construction in \cref{thm:replicated-cost} attains $D=R=m$.
The matching lower bounds are
\cref{thm:depth-lower-bound,thm:rotation-lower-bound}.  Therefore neither
resource can be reduced in the stated model.
\end{proof}

The result is an exact statement about monoid depth and counted online
rotations, not a claim that the replicated scan is faster for every
homomorphic backend.  Relative to the direct transported-predecessor scan,
it reduces rotations from
$m\cdot(m+1)/2$ to $m$, but increases packed monoid work from $m$ to
$2\cdot m-1$ and requires a second persistent state.  The concrete crossover
depends on the ciphertext cost of one monoid composition, key switching,
scale management, and available memory; these quantities are evaluated only
after the carry representation is fixed.

\section{Rotation-Key Frontier}
\label{sec:key-frontier}

The optimal scan of \cref{sec:replicated-scan} uses one effective rotation at
 each of the $m$ two-adic scales.  If the evaluator stores a direct rotation
 key for every required offset, these effective rotations cost exactly $m$
 online calls.  Reducing the directly supported key set saves evaluation-key
 material, but an omitted scale must then be synthesized by a sequence of
 directly keyed rotations.  This section determines the resulting lower and
 upper bounds and proves an exact frontier whenever they coincide.

The distinction between \emph{directly keyed offsets} and \emph{online
calls} is essential.  A key budget $K$ limits the number of distinct offsets
that can be invoked directly; it does not limit how often those offsets may
be used.  The theorem below counts every invocation, including each call in a
sequence synthesizing a larger effective rotation.  It is therefore a
program-specific online-call bound, rather than a byte-level estimate of
backend key storage.  Generic systems may further trade transmitted, stored,
or derived evaluation keys against runtime~\cite{LeeLKNo23}; any offset that
is directly available to the online evaluator belongs to the effective key
set counted here.

Throughout this section, let $1\leq K\leq m$.  We write
$R^\star_{\mathrm{pref}}(m,K)$ for the minimum number of online
cyclic-rotation invocations among packed circuits that compute all ordered
prefixes on $n=2^m$ slots while using at most $K$ distinct directly keyed
nonzero offsets.  The lower bounds below do not assume minimum monoid depth;
the matching constructions retain the optimal depth $m$.

\subsection{Coefficient-box lower bound}
\label{sec:key-product-lower}

Suppose the directly keyed offsets are
$g_1,\ldots,g_K\in\Z_{2^m}$, and offset $g_j$ is invoked $c_j$ times.  We may
pad a smaller key set with zero invocation counts, so the notation also
covers circuits using fewer than $K$ keys.

\begin{theorem}[Direct-key product bound]
\label{thm:key-product-bound}
If one output depends on all $2^m$ input slots, then
\begin{equation}
  \prod_{j=1}^{K}(c_j+1)\geq 2^m.
  \label{eq:key-product-bound}
\end{equation}
Consequently, every circuit computing all inclusive or all exclusive prefixes
satisfies \cref{eq:key-product-bound}.
\end{theorem}

\begin{proof}
Fix an output slot $p$.  Expand the packed circuit into the slot-labelled
acyclic graph used in the proof of
\cref{lem:subset-displacement}.  Along a dependency path ending at $p$, let
$a_j$ be the number of crossed rotation-gate invocations whose directly keyed
offset is $g_j$.  Since the complete circuit contains exactly $c_j$ such
invocations,
\begin{equation}
  0\leq a_j\leq c_j
  \qquad\text{for every }1\leq j\leq K.
  \label{eq:key-coefficient-range}
\end{equation}
Local edges preserve the physical slot label, while every crossed invocation
of key $g_j$ adds $g_j$.  Therefore the source displacement of the path is
of the form
\begin{equation}
  \sum_{j=1}^{K}a_j\cdot g_j
  \pmod {2^m}.
  \label{eq:key-coefficient-displacement}
\end{equation}
The coefficient vector $(a_1,\ldots,a_K)$ belongs to the box
\begin{equation}
  \prod_{j=1}^{K}\{0,\ldots,c_j\},
  \label{eq:key-coefficient-box}
\end{equation}
which has cardinality $\prod_{j=1}^{K}(c_j+1)$.  Distinct coefficient vectors
may yield the same residue in \cref{eq:key-coefficient-displacement}, but no
path can yield a displacement outside this image.  Hence at most
$\prod_{j=1}^{K}(c_j+1)$ distinct input slots can affect the fixed output.
If that output depends on all $2^m$ inputs, the image must contain every
residue, proving \cref{eq:key-product-bound}.

For an inclusive scan, the last logical prefix depends on all inputs.  For an
exclusive scan with $m\geq2$, the last prefix depends on $2^m-1$ inputs.  If
the product in \cref{eq:key-product-bound} were smaller than $2^m$, then,
because it is an integer product of positive integers, the strongest possible
conclusion for the exclusive output would only be a bound by $2^m-1$.
However, the complete carry-out or the inclusive aggregate obtained locally
from the final exclusive state and the last input depends on all $2^m$
inputs and uses no additional rotation.  Applying the first part to this
locally derived output proves the same product bound for the exclusive scan.
\end{proof}

The product bound can be converted into the strongest lower bound depending
only on $m$ and $K$.  The required optimization is discrete: for a fixed
number of calls, the coefficient box is largest when the invocation counts
are as balanced as possible.

\begin{lemma}[Balanced coefficient box]
\label{lem:key-balanced-box}
Let $c_1,\ldots,c_K\in\Z_{\geq0}$ satisfy
\begin{equation}
  \sum_{j=1}^{K}c_j=R.
  \label{eq:key-count-sum}
\end{equation}
Write
\begin{equation}
  R=a\cdot K+b,
  \qquad
  0\leq b<K.
  \label{eq:key-R-division}
\end{equation}
Then
\begin{equation}
  \prod_{j=1}^{K}(c_j+1)
  \leq
  (a+1)^{K-b}\cdot(a+2)^b.
  \label{eq:key-balanced-product}
\end{equation}
Equality holds if and only if, up to reordering, $K-b$ counts equal $a$ and
$b$ counts equal $a+1$.
\end{lemma}

\begin{proof}
Suppose two counts satisfy $c_u\geq c_v+2$.  Move one invocation from the
larger count to the smaller one.  The factors involving these two counts
change from
\begin{equation}
  (c_u+1)\cdot(c_v+1)
  \label{eq:key-unbalanced-factors}
\end{equation}
to
\begin{equation}
  c_u\cdot(c_v+2).
  \label{eq:key-balanced-factors}
\end{equation}
Their difference is
\begin{align}
  c_u\cdot(c_v+2)-(c_u+1)\cdot(c_v+1)
  &=c_u-c_v-1 \\
  &>0.
  \label{eq:key-balancing-increase}
\end{align}
Thus any pair of counts differing by at least two can be balanced while
strictly increasing the product and preserving the sum.  Repeating this
operation terminates only when all counts differ by at most one.  Under
\cref{eq:key-R-division}, the unique multiset with that property contains
$K-b$ copies of $a$ and $b$ copies of $a+1$.  Substituting these values yields
\cref{eq:key-balanced-product}.  Strict increase at every nontrivial
balancing step also proves the equality characterization.
\end{proof}

\begin{theorem}[Exact integer counting lower bound]
\label{thm:key-integer-lower}
Every packed prefix circuit using at most $K$ directly keyed offsets satisfies
\begin{equation}
  R^\star_{\mathrm{pref}}(m,K)\geq L(m,K),
  \label{eq:key-L-lower}
\end{equation}
where
\begin{equation}
  L(m,K)
  \defeq
  \min_{\substack{a\in\Z_{\geq0},\ 0\leq b<K}}
  \left\{
    a\cdot K+b:
    (a+1)^{K-b}\cdot(a+2)^b\geq2^m
  \right\}.
  \label{eq:key-L-definition}
\end{equation}
In particular,
\begin{equation}
  R^\star_{\mathrm{pref}}(m,K)
  \geq
  \left\lceil
    K\cdot\bigl(2^{m/K}-1\bigr)
  \right\rceil.
  \label{eq:key-amgm-lower}
\end{equation}
\end{theorem}

\begin{proof}
Let a circuit use $R$ calls and write $R=a\cdot K+b$ as in
\cref{eq:key-R-division}.  By \cref{thm:key-product-bound}, its invocation
counts satisfy
\begin{equation}
  2^m\leq\prod_{j=1}^{K}(c_j+1).
  \label{eq:key-product-lower-proof}
\end{equation}
By \cref{lem:key-balanced-box},
\begin{equation}
  \prod_{j=1}^{K}(c_j+1)
  \leq
  (a+1)^{K-b}\cdot(a+2)^b.
  \label{eq:key-product-upper-proof}
\end{equation}
Hence the pair $(a,b)$ associated with $R$ is feasible in
\cref{eq:key-L-definition}, and therefore $R\geq L(m,K)$.

For \cref{eq:key-amgm-lower}, apply the arithmetic--geometric mean inequality
to the $K$ positive integers $c_j+1$:
\begin{equation}
  \left(\prod_{j=1}^{K}(c_j+1)\right)^{1/K}
  \leq
  \frac{1}{K}\cdot\sum_{j=1}^{K}(c_j+1)
  =1+\frac{R}{K}.
  \label{eq:key-amgm-step}
\end{equation}
Combining \cref{eq:key-product-bound,eq:key-amgm-step} gives
$2^{m/K}\leq1+R/K$, and therefore
$R\geq K\cdot(2^{m/K}-1)$.  Since $R$ is an integer, taking the ceiling gives
\cref{eq:key-amgm-lower}.
\end{proof}

\subsection{Balanced dyadic construction}
\label{sec:key-balanced-construction}

The lower bound is matched in many parameter regimes by grouping consecutive
two-adic scales.  Write
\begin{equation}
  m=q\cdot K+s,
  \qquad
  0\leq s<K.
  \label{eq:key-m-division}
\end{equation}
Choose $K-s$ group lengths equal to $q$ and $s$ group lengths equal to
$q+1$.  Arrange these groups consecutively so that they partition the
valuation set $\{0,\ldots,m-1\}$.  If one group begins at valuation $v$ and
has length $\alpha$, store a direct key only for offset $2^v$.

The scan level at valuation $v+r$, where $0\leq r<\alpha$, requires the
effective rotation $2^{v+r}$.  Repeating the directly keyed rotation
$2^v$ exactly $2^r$ times synthesizes this offset, because cyclic rotations
compose additively:
\begin{equation}
  \underbrace{\rho_{2^v}\circ\cdots\circ\rho_{2^v}}_{2^r\text{ calls}}
  =\rho_{2^r\cdot2^v}
  =\rho_{2^{v+r}}.
  \label{eq:key-repeated-rotation}
\end{equation}
The total number of calls contributed by this group is therefore
\begin{equation}
  \sum_{r=0}^{\alpha-1}2^r=2^\alpha-1.
  \label{eq:key-group-cost}
\end{equation}

\begin{proposition}[Balanced dyadic upper bound]
\label{prop:key-balanced-upper}
For $m=q\cdot K+s$ with $0\leq s<K$, the replicated scan can be evaluated
with $K$ directly keyed offsets, monoid depth $m$, and
\begin{equation}
  U(m,K)
  \defeq
  (K-s)\cdot(2^q-1)
  +s\cdot(2^{q+1}-1)
  \label{eq:key-U-definition}
\end{equation}
online rotation calls.  Consequently,
\begin{equation}
  R^\star_{\mathrm{pref}}(m,K)\leq U(m,K).
  \label{eq:key-U-upper}
\end{equation}
\end{proposition}

\begin{proof}
Partition the $m$ scan valuations into the consecutive groups described
above and directly key the first offset of each group.  By
\cref{eq:key-repeated-rotation}, every effective scale rotation required by
the scan can be synthesized using only its group's direct key.  A group of
length $q$ costs $2^q-1$ calls by \cref{eq:key-group-cost}, and a group of
length $q+1$ costs $2^{q+1}-1$ calls.  Summing over $K-s$ short groups and
$s$ long groups gives \cref{eq:key-U-definition}.

Replacing one effective rotation by a sequence of rotations changes only the
routing performed before that scan level.  It does not add a monoid gate to
any dependency path.  Therefore the construction retains monoid depth $m$
and the correctness invariant of \cref{thm:replicated-exclusive-correctness}.
It uses exactly one directly keyed offset per group, hence at most $K$ direct
keys.
\end{proof}

Combining the lower and upper bounds gives the general proved frontier
sandwich
\begin{equation}
  L(m,K)
  \leq
  R^\star_{\mathrm{pref}}(m,K)
  \leq
  U(m,K).
  \label{eq:key-frontier-sandwich}
\end{equation}

\begin{corollary}[Exact coincidence cases]
\label{cor:key-coincidence}
If $L(m,K)=U(m,K)$, then
\begin{equation}
  R^\star_{\mathrm{pref}}(m,K)=L(m,K)=U(m,K).
  \label{eq:key-coincidence}
\end{equation}
\end{corollary}

\begin{proof}
The result follows immediately by substituting the assumed equality into
\cref{eq:key-frontier-sandwich}.
\end{proof}

\begin{table}[t]
  \centering
  \caption{Representative direct-key parameter pairs.  Equality of $L$ and
  $U$ proves the exact prefix frontier.  A strict gap records only the current
  proved interval, not evidence of suboptimality of the dyadic construction.}
  \label{tab:key-frontier-examples}
  \setlength{\tabcolsep}{7pt}
  \renewcommand{\arraystretch}{1.18}
  \begin{tabular}{@{}ccrrl@{}}
    \toprule
    $m$ & $K$ & $L(m,K)$ & $U(m,K)$ & proved status \\
    \midrule
    $5$  & $3$ & $7$  & $7$  & exact \\
    $8$  & $3$ & $17$ & $17$ & exact \\
    $8$  & $5$ & $11$ & $11$ & exact \\
    $7$  & $2$ & $21$ & $22$ & $21\leq R^\star_{\mathrm{pref}}\leq22$ \\
    $10$ & $3$ & $28$ & $29$ & $28\leq R^\star_{\mathrm{pref}}\leq29$ \\
    \bottomrule
  \end{tabular}
\end{table}

\subsection{Exact divisible frontier and equality rigidity}
\label{sec:key-divisible}

When $K$ divides $m$, the smooth lower bound in
\cref{eq:key-amgm-lower} is integral and exactly matches the dyadic
construction.

\begin{theorem}[Exact frontier for $K\mid m$]
\label{thm:key-divisible-frontier}
Assume $K\mid m$ and set
\begin{equation}
  \alpha\defeq\frac{m}{K}.
  \label{eq:key-alpha-definition}
\end{equation}
Then
\begin{equation}
  R^\star_{\mathrm{pref}}(m,K)
  =K\cdot\bigl(2^\alpha-1\bigr)
  =K\cdot\bigl(2^{m/K}-1\bigr).
  \label{eq:key-divisible-frontier}
\end{equation}
The value is attained at monoid depth $m$ by directly keying the offsets
\begin{equation}
  2^{0\cdot\alpha},
  2^{1\cdot\alpha},
  \ldots,
  2^{(K-1)\cdot\alpha}.
  \label{eq:key-divisible-key-set}
\end{equation}
\end{theorem}

\begin{proof}
The lower bound \cref{eq:key-amgm-lower} gives
\begin{equation}
  R^\star_{\mathrm{pref}}(m,K)
  \geq
  K\cdot\bigl(2^{m/K}-1\bigr)
  =K\cdot\bigl(2^\alpha-1\bigr).
  \label{eq:key-divisible-lower-proof}
\end{equation}
In \cref{eq:key-m-division}, divisibility means $q=\alpha$ and $s=0$.
The upper bound of \cref{prop:key-balanced-upper} therefore becomes
\begin{equation}
  U(m,K)=K\cdot(2^\alpha-1).
  \label{eq:key-divisible-upper-proof}
\end{equation}
The construction partitions the valuation range into $K$ consecutive groups
of length $\alpha$ and uses exactly the key set in
\cref{eq:key-divisible-key-set}.  The matching lower and upper bounds prove
\cref{eq:key-divisible-frontier}.
\end{proof}

The equality case determines not only the number of calls, but also the
valuation profile of every optimal direct-key set.

\begin{theorem}[Mixed-radix two-adic rigidity]
\label{thm:key-divisible-rigidity}
Let $K\mid m$, let $\alpha=m/K$, and suppose a packed circuit uses at most
$K$ directly keyed offsets and exactly
\begin{equation}
  R=K\cdot(2^\alpha-1)
  \label{eq:key-rigidity-R}
\end{equation}
rotation calls to produce one output depending on all $2^m$ inputs.  Then:
\begin{enumerate}[leftmargin=*,itemsep=0.25em,topsep=0.25em]
  \item exactly $K$ direct keys are used, and each is invoked
  \begin{equation}
    c_j=2^\alpha-1
    \label{eq:key-rigidity-counts}
  \end{equation}
  times;
  \item after reordering the directly keyed offsets,
  \begin{equation}
    \nuTwo(g_j)=j\cdot\alpha
    \qquad\text{for }0\leq j<K.
    \label{eq:key-rigidity-valuations}
  \end{equation}
\end{enumerate}
Thus every optimal key set begins one complete block of $\alpha$ consecutive
two-adic scales.
\end{theorem}

\begin{proof}
Pad the circuit's key set to $K$ offsets by assigning zero invocation counts
to any unused positions.  By \cref{thm:key-product-bound},
\begin{equation}
  2^m\leq\prod_{j=1}^{K}(c_j+1).
  \label{eq:key-rigidity-product-lower}
\end{equation}
On the other hand, the arithmetic--geometric mean inequality and
\cref{eq:key-rigidity-R} give
\begin{align}
  \prod_{j=1}^{K}(c_j+1)
  &\leq
  \left(
    \frac{1}{K}\cdot\sum_{j=1}^{K}(c_j+1)
  \right)^K \\
  &=
  \left(
    1+\frac{K\cdot(2^\alpha-1)}{K}
  \right)^K \\
  &=2^{\alpha\cdot K}
   =2^m.
  \label{eq:key-rigidity-product-upper}
\end{align}
Hence equality holds throughout.  Equality in arithmetic--geometric mean
forces all $c_j+1$ to be equal to $2^\alpha$, proving
\cref{eq:key-rigidity-counts}.  In particular, no padded count is zero, so
exactly $K$ keys are used.

Fix an output slot depending on all inputs.  A dependency path determines a
coefficient vector
\begin{equation}
  (a_0,\ldots,a_{K-1})
  \in\{0,\ldots,2^\alpha-1\}^K,
  \label{eq:key-rigidity-coefficient-box}
\end{equation}
and the associated source displacement is
$\sum_j a_j\cdot g_j$ modulo $2^m$.  Every residue must occur because the output depends on all input slots.
More precisely, the coefficient vectors realized by dependency paths form a
subset of the box in \cref{eq:key-rigidity-coefficient-box}, and their image
already contains all $2^m$ residues.  The complete box has cardinality
\begin{equation}
  (2^\alpha)^K=2^{\alpha\cdot K}=2^m,
  \label{eq:key-rigidity-box-cardinality}
\end{equation}
which equals the codomain size.  The path-realized subset must therefore be
the entire box, and the coefficient map from the box to $\Z_{2^m}$ is
bijective.
Equivalently, in the quotient ring $\Z[X]/(X^{2^m}-1)$,
\begin{equation}
  \prod_{j=0}^{K-1}
  \left(
    \sum_{a=0}^{2^\alpha-1}X^{a\cdot g_j}
  \right)
  =
  \sum_{r=0}^{2^m-1}X^r.
  \label{eq:key-rigidity-group-ring}
\end{equation}
Indeed, the coefficient of $X^r$ on the left counts coefficient vectors
mapping to residue $r$, and bijectivity makes every such coefficient equal to
one.

For $0\leq k<m$, let $\zeta_k$ be a primitive $2^{k+1}$-st root of unity.
It is also a $2^m$-th root of unity, and evaluating the right-hand side of
\cref{eq:key-rigidity-group-ring} gives zero.  Hence, for every $k$, at least
one geometric factor on the left vanishes.  Let $v_j=\nuTwo(g_j)$.  We claim
that
\begin{equation}
  \sum_{a=0}^{2^\alpha-1}\zeta_k^{a\cdot g_j}=0
  \quad\Longleftrightarrow\quad
  v_j\leq k<v_j+\alpha.
  \label{eq:key-rigidity-factor-zero}
\end{equation}
To prove the claim, set $z=\zeta_k^{g_j}$.  If $v_j>k$, then $z=1$ and the
sum equals $2^\alpha$, so it does not vanish.  If $v_j\leq k$, then $z\neq1$
and $z$ has order $2^{k+1-v_j}$.  The geometric sum vanishes exactly when
$z^{2^\alpha}=1$, which is equivalent to
$2^{k+1-v_j}\mid2^\alpha$, or $k<v_j+\alpha$.  This proves
\cref{eq:key-rigidity-factor-zero}.

Thus the $K$ integer intervals
\begin{equation}
  I_j\defeq\{v_j,v_j+1,\ldots,v_j+\alpha-1\}
  \label{eq:key-rigidity-intervals}
\end{equation}
cover every valuation $k\in\{0,\ldots,m-1\}$.  Each interval has exactly
$\alpha$ elements, so the total number of interval positions counted with
multiplicity is
\begin{equation}
  K\cdot\alpha=m.
  \label{eq:key-rigidity-total-length}
\end{equation}
Covering a set of $m$ valuations with total multiplicity $m$ leaves no room
for overlap or for positions outside $\{0,\ldots,m-1\}$.  Consequently, the
$I_j$ are pairwise disjoint and partition the complete valuation range.
After ordering them by their first elements, the only partition into
consecutive intervals of common length $\alpha$ is
\begin{equation}
  \{0,\ldots,\alpha-1\},
  \{\alpha,\ldots,2\cdot\alpha-1\},
  \ldots,
  \{(K-1)\cdot\alpha,\ldots,K\cdot\alpha-1\}.
  \label{eq:key-rigidity-partition}
\end{equation}
Their starting valuations are therefore
$0,\alpha,\ldots,(K-1)\cdot\alpha$, proving
\cref{eq:key-rigidity-valuations}.
\end{proof}

\subsection{Nondivisible residual and scope}
\label{sec:key-nondivisible}

For $K\nmid m$, the proved statement is exactly the sandwich in
\cref{eq:key-frontier-sandwich}.  The bounds coincide for many nondivisible
pairs, as illustrated in \cref{tab:key-frontier-examples}, but not for all of
them.  For example,
\begin{equation}
  21\leq R^\star_{\mathrm{pref}}(7,2)\leq22,
  \qquad
  28\leq R^\star_{\mathrm{pref}}(10,3)\leq29.
  \label{eq:key-open-examples}
\end{equation}
The lower endpoint certifies only that enough coefficient combinations are
available to reach every displacement.  It does not by itself enforce the
causal order, all-output semantics, or the $m$ successive sibling exchanges
of an ordered prefix scan.  Conversely, the upper endpoint is an explicit
minimum-depth scan, not a claim of optimality when $L(m,K)<U(m,K)$.

The remaining nondivisible gaps are therefore left open.  Computational
searches may support the balanced dyadic construction for selected
parameters, but such evidence is not used in any theorem of this paper.  The
radix carry and borrow instantiations in the next section may select a
budget dividing $m$, where the direct-key frontier and its equality structure
are fully determined by
\cref{thm:key-divisible-frontier,thm:key-divisible-rigidity}.

\section{Radix-\texorpdfstring{\CKKS}{CKKS} Instantiation}
\label{sec:radix-ckks}

This section specializes the abstract exclusive scan to radix carry and borrow.
The algebraic statements are exact: they assume that every provisional digit is
classified into the correct transition state and that the resulting state
arithmetic is decoded without error.  A concrete \CKKS ciphertext carries
approximate values, so the numerical conditions under which classification and
rounding recover these exact states are deferred to the implementation and
precision analysis.  This separation is useful because the scan topology is
independent of the selected bootstrapping and state-classification backend.

\subsection{A one-ciphertext carry-state encoding}
\label{sec:radix-symbols}

The transition monoid from \cref{sec:model-carry} has only three elements:
kill, propagate, and generate.  Following the compressed symbolic realization
used for logarithmic radix carry~\cite{ChaPL26}, encode them by
\begin{equation}
  \sigma(\mathsf K)=0,
  \qquad
  \sigma(\mathsf P)=\frac{1}{2},
  \qquad
  \sigma(\mathsf G)=\mathrm{i},
  \label{eq:radix-symbol-encoding}
\end{equation}
where $\mathrm{i}^2=-1$.  If $x$ encodes a lower interval and $y$ encodes an
adjacent upper interval, define
\begin{equation}
  x\circledast y
  \defeq
  y+(y+\overline{y})\cdot(x-y).
  \label{eq:radix-symbol-composition}
\end{equation}
The operand order matches the convention of \cref{eq:model-transition-order}:
the lower transition acts first.

\begin{proposition}[Symbolic transition isomorphism]
\label{prop:radix-symbol-isomorphism}
For all transition states $A,X\in\{\mathsf K,\mathsf P,\mathsf G\}$,
\begin{equation}
  \sigma(A\diamond X)
  =\sigma(A)\circledast\sigma(X).
  \label{eq:radix-symbol-homomorphism}
\end{equation}
The identity symbol is $1/2$.  One packed evaluation of
\cref{eq:radix-symbol-composition} uses one ciphertext--ciphertext
multiplication, additions, and one complex conjugation when the state is
stored in one \CKKS ciphertext.
\end{proposition}

\begin{proof}
Let $x=\sigma(A)$ and $y=\sigma(X)$.  The result is determined by the upper
state $X$.
If $X=\mathsf K$, then $y=0$ and
$x\circledast y=0=\sigma(\mathsf K)$; an upper kill state discards the lower
transition.  If $X=\mathsf P$, then $y=1/2=\overline y$, and therefore
\begin{equation}
  x\circledast y
  =\frac{1}{2}+1\cdot\left(x-\frac{1}{2}\right)
  =x;
  \label{eq:radix-symbol-propagate-case}
\end{equation}
this is exactly the identity action of an upper propagate state.  Finally, if
$X=\mathsf G$, then $y=\mathrm{i}$ and $y+\overline y=0$, so
$x\circledast y=\mathrm{i}=\sigma(\mathsf G)$; an upper generate state also
discards the lower transition.  These three cases are precisely the
composition table of the transition monoid.  The propagate symbol $1/2$ is
therefore its identity.

Equation~\eqref{eq:radix-symbol-composition} contains one product of the two
ciphertext-dependent quantities $y+\overline y$ and $x-y$.  All remaining
operations are linear, apart from the conjugation automorphism.
\end{proof}

For a provisional digit $z\in\{0,\ldots,2\cdot B-2\}$, define the exact
state classifier
\begin{equation}
  \chi_B(z)
  \defeq
  \begin{cases}
    0, & 0\leq z\leq B-2,\\
    \frac{1}{2}, & z=B-1,\\
    \mathrm{i}, & B\leq z\leq2\cdot B-2.
  \end{cases}
  \label{eq:radix-carry-classifier}
\end{equation}
Thus $\chi_B(z)=\sigma(F_z)$.  In a concrete discrete-\CKKS realization,
$\chi_B$ is evaluated slotwise by a functional bootstrap or another
encrypted classifier.  The prefix circuit starts after this classification
step.

Two linear functions recover the generate and propagate indicators of a
symbol $x\in\{0,1/2,\mathrm{i}\}$:
\begin{equation}
  \operatorname{gen}(x)
  \defeq
  \frac{x-\overline{x}}{2\cdot\mathrm{i}},
  \qquad
  \operatorname{prop}(x)
  \defeq
  x+\overline{x}.
  \label{eq:radix-symbol-indicators}
\end{equation}
They take values in $\{0,1\}$ and satisfy
\begin{equation}
  F_x(c)
  =\operatorname{gen}(x)
   +\operatorname{prop}(x)\cdot c
  \qquad\text{for }c\in\{0,1\}.
  \label{eq:radix-symbol-state-on-bit}
\end{equation}
Here $F_x$ denotes the transition represented by $x$.  In particular,
$F_x(0)=\operatorname{gen}(x)$.

\subsection{Exclusive carry and canonical digits}
\label{sec:radix-carry}

Let $z_0,\ldots,z_{n-1}$ be provisional radix-$B$ digits satisfying
\begin{equation}
  0\leq z_i\leq2\cdot B-2,
  \qquad
  n=2^m,
  \label{eq:radix-reduced-digit-range}
\end{equation}
and set $c_0=0$.  Sequential carry propagation is
\begin{equation}
  c_{i+1}
  \defeq
  \left\lfloor\frac{z_i+c_i}{B}\right\rfloor,
  \qquad
  d_i
  \defeq
  z_i+c_i-B\cdot c_{i+1}.
  \label{eq:radix-sequential-carry}
\end{equation}
The range in \cref{eq:radix-reduced-digit-range} ensures that every carry is
binary.

Let $S_i$ be the transition state classified from $z_i$, and let
$E_i=S_0\diamond\cdots\diamond S_{i-1}$ be its exclusive prefix, with
$E_0=e$.  The replicated scan of \cref{sec:replicated-scan} returns all
$E_i$ directly in bit-reversed order.

\begin{theorem}[Exact exclusive carry]
\label{thm:radix-exclusive-carry}
Assume that the state classifier in \cref{eq:radix-carry-classifier} is
correct.  For every $i\in\{0,\ldots,n-1\}$, define
\begin{equation}
  c_i=F_{E_i}(0),
  \qquad
  c_{i+1}=F_{S_i}(c_i),
  \qquad
  d_i=z_i+c_i-B\cdot c_{i+1}.
  \label{eq:radix-exclusive-carry-evaluation}
\end{equation}
Then:
\begin{enumerate}[leftmargin=*,itemsep=0.25em,topsep=0.25em]
  \item the values $c_i$ equal the sequential carries of
  \cref{eq:radix-sequential-carry} and lie in $\{0,1\}$;
  \item every corrected digit satisfies $0\leq d_i\leq B-1$;
  \item the corrected representation preserves the integer value:
  \begin{equation}
    \sum_{i=0}^{n-1}z_i\cdot B^i
    =
    \sum_{i=0}^{n-1}d_i\cdot B^i
    +c_n\cdot B^n;
    \label{eq:radix-carry-value-identity}
  \end{equation}
  \item in one complete bit-reversed cyclic domain, all carry-in states are
  obtained with $m$ rotations, monoid depth $m$, and
  $2\cdot m-1$ packed symbolic compositions.  No final logical-predecessor
  shift is required.
\end{enumerate}
\end{theorem}

\begin{proof}
By definition, $E_0=e$, so $F_{E_0}(0)=0=c_0$.  Assume inductively that
$F_{E_i}(0)$ equals the sequential carry entering digit $i$.  Ordered
composition gives
\begin{equation}
  E_{i+1}=E_i\diamond S_i,
  \qquad
  F_{E_{i+1}}=F_{S_i}\circ F_{E_i}.
  \label{eq:radix-carry-prefix-recurrence}
\end{equation}
Hence
\begin{equation}
  F_{E_{i+1}}(0)
  =F_{S_i}\bigl(F_{E_i}(0)\bigr)
  =F_{S_i}(c_i)
  =\left\lfloor\frac{z_i+c_i}{B}\right\rfloor
  =c_{i+1}.
  \label{eq:radix-carry-induction}
\end{equation}
This proves equality with sequential propagation for every $i$.

Because $0\leq z_i\leq2\cdot B-2$ and $c_i\in\{0,1\}$,
\begin{equation}
  0\leq z_i+c_i\leq2\cdot B-1<2\cdot B.
  \label{eq:radix-carry-sum-range}
\end{equation}
Therefore its quotient by $B$ is either zero or one, proving
$c_{i+1}\in\{0,1\}$.  The value $d_i$ in
\cref{eq:radix-exclusive-carry-evaluation} is the Euclidean remainder of
$z_i+c_i$ modulo $B$, and thus lies in $\{0,\ldots,B-1\}$.

Multiplying the digit equation by $B^i$ and summing gives
\begin{align}
  \sum_{i=0}^{n-1}d_i\cdot B^i
  &=
  \sum_{i=0}^{n-1}z_i\cdot B^i
  +\sum_{i=0}^{n-1}c_i\cdot B^i
  -\sum_{i=0}^{n-1}c_{i+1}\cdot B^{i+1}\\
  &=
  \sum_{i=0}^{n-1}z_i\cdot B^i
  +c_0-c_n\cdot B^n.
  \label{eq:radix-carry-telescoping}
\end{align}
Since $c_0=0$, rearrangement yields
\cref{eq:radix-carry-value-identity}.

Finally, \cref{thm:replicated-cost} gives the scan costs.
The scan places $E_i$ at the same physical slot as digit $i$.  Both
$F_{E_i}(0)$ and $F_{S_i}(c_i)$ are therefore slotwise evaluations, so no
post-scan routing is needed.  In the one-ciphertext symbol representation,
one abstract state rotation is one ciphertext rotation.
\end{proof}

If $e_i=\sigma(E_i)$ and $s_i=\sigma(S_i)$, the local evaluation can be
written explicitly as
\begin{equation}
  c_i=\operatorname{gen}(e_i),
  \qquad
  c_{i+1}
  =\operatorname{gen}(s_i)
   +\operatorname{prop}(s_i)\cdot c_i.
  \label{eq:radix-symbol-carry-evaluation}
\end{equation}
The first expression is linear in $e_i$ and $\overline{e_i}$.  The second
uses one additional slotwise ciphertext multiplication after the scan.  It
replaces the final predecessor rotation required by an inclusive-only
formulation.  Thus, after state classification, the exclusive symbolic
realization has scan multiplicative depth $m$ and carry-correction depth at
most $m+1$; exact scale and level accounting is backend dependent.

\subsection{Borrow, comparison, and conditional subtraction}
\label{sec:radix-borrow}

Let $x_i,y_i\in\{0,\ldots,B-1\}$ be canonical digits.  The outgoing borrow
from digit $i$ is the Boolean transition
\begin{equation}
  G_{x_i,y_i}(b)
  \defeq
  \begin{cases}
    1, & x_i-y_i-b<0,\\
    0, & x_i-y_i-b\geq0,
  \end{cases}
  \qquad b\in\{0,1\}.
  \label{eq:radix-borrow-transition}
\end{equation}
It is a kill state when $x_i>y_i$, a propagate state when $x_i=y_i$, and a
generate state when $x_i<y_i$.  Hence it is represented by the same monoid
and the same complex symbols as carry.

\begin{theorem}[Borrow and comparison]
\label{thm:radix-borrow-comparison}
Let
\begin{equation}
  X\defeq\sum_{i=0}^{n-1}x_i\cdot B^i,
  \qquad
  Y\defeq\sum_{i=0}^{n-1}y_i\cdot B^i,
  \label{eq:radix-integers-xy}
\end{equation}
set $b_0=0$, and let $E_i^{-}$ be the exclusive prefix of the borrow
transitions in \cref{eq:radix-borrow-transition}.  Define
\begin{equation}
  b_i=F_{E_i^{-}}(0),
  \qquad
  b_{i+1}=G_{x_i,y_i}(b_i),
  \qquad
  q_i=x_i-y_i-b_i+B\cdot b_{i+1}.
  \label{eq:radix-borrow-evaluation}
\end{equation}
Then $q_i\in\{0,\ldots,B-1\}$ and
\begin{equation}
  \sum_{i=0}^{n-1}q_i\cdot B^i
  =X-Y+b_n\cdot B^n.
  \label{eq:radix-borrow-value-identity}
\end{equation}
Moreover,
\begin{equation}
  b_n=1
  \quad\Longleftrightarrow\quad
  X<Y.
  \label{eq:radix-final-borrow-comparison}
\end{equation}
The borrow scan uses the same $m$ rotations as carry.
\end{theorem}

\begin{proof}
The proof that $b_i=F_{E_i^{-}}(0)$ equals sequential borrow propagation is
identical to the induction in
\cref{eq:radix-carry-prefix-recurrence,eq:radix-carry-induction}, with
$G_{x_i,y_i}$ in place of $F_{S_i}$.

If $b_{i+1}=0$, then $x_i-y_i-b_i\geq0$, and this quantity is at most
$B-1$; hence $q_i=x_i-y_i-b_i\in\{0,\ldots,B-1\}$.  If
$b_{i+1}=1$, then
\begin{equation}
  -B\leq x_i-y_i-b_i\leq-1,
  \label{eq:radix-negative-difference-range}
\end{equation}
so adding $B$ again yields $q_i\in\{0,\ldots,B-1\}$.

Multiplying \cref{eq:radix-borrow-evaluation} by $B^i$ and summing telescopes:
\begin{align}
  \sum_{i=0}^{n-1}q_i\cdot B^i
  &=X-Y
    -\sum_{i=0}^{n-1}b_i\cdot B^i
    +\sum_{i=0}^{n-1}b_{i+1}\cdot B^{i+1}\\
  &=X-Y-b_0+b_n\cdot B^n,
  \label{eq:radix-borrow-telescoping}
\end{align}
which gives \cref{eq:radix-borrow-value-identity} because $b_0=0$.
The left-hand side lies in $[0,B^n-1]$.  If $b_n=0$, then
$X-Y\geq0$.  If $b_n=1$, then
\begin{equation}
  X-Y
  =\sum_{i=0}^{n-1}q_i\cdot B^i-B^n
  \leq-1.
  \label{eq:radix-borrow-sign}
\end{equation}
This proves \cref{eq:radix-final-borrow-comparison}.  The routing cost is the
same because the transition monoid and scan topology are unchanged.
\end{proof}

Comparison needs only the final borrow at the most significant digit.  A
conditional operation on every digit instead needs a replicated copy of this
bit.  The scan can provide it without another rotation by retaining the final
aggregate update that was omitted from the $2\cdot m-1$ work count.

\begin{proposition}[Conditional subtraction without additional routing]
\label{prop:radix-conditional-subtraction}
Let $M$ be an $n$-digit positive integer, let $0\leq X<2\cdot M$, and compute
the borrow digits $q_i$ of $X-M$ as in
\cref{thm:radix-borrow-comparison}.  If the final complete aggregate
$A^{(m)}$ is retained, its transition is replicated at every slot and
therefore yields the final borrow $b_n$ at every digit.  The slotwise
selection
\begin{equation}
  r_i
  \defeq
  b_n\cdot x_i+(1-b_n)\cdot q_i
  \label{eq:radix-conditional-selection}
\end{equation}
encodes
\begin{equation}
  R
  =
  \begin{cases}
    X, & X<M,\\
    X-M, & X\geq M,
  \end{cases}
  \qquad 0\leq R<M.
  \label{eq:radix-conditional-subtraction-result}
\end{equation}
It uses $m$ rotations, monoid depth $m$, and $2\cdot m$ packed monoid
compositions.
\end{proposition}

\begin{proof}
By \cref{eq:radix-final-borrow-comparison}, $b_n=1$ exactly when $X<M$.
In that case \cref{eq:radix-conditional-selection} selects the original digit
$x_i$.  Otherwise it selects $q_i$, and
\cref{eq:radix-borrow-value-identity} with $b_n=0$ shows that these digits
encode $X-M$.  The assumption $X<2\cdot M$ places either selected value in
$[0,M-1]$.

At every nonfinal level, the exclusive scan already computes both the next
aggregate and the next prefix.  At the final level, computing the aggregate
as well as the prefix adds one packed monoid composition in parallel.  It
therefore changes $T$ from $2\cdot m-1$ to $2\cdot m$ without changing
monoid depth or rotation count.  By the aggregate invariant,
$A^{(m)}$ is the complete transition replicated throughout the cyclic
domain, so $F_{A^{(m)}}(0)=b_n$ is available at every digit without a
broadcast rotation.
\end{proof}

\subsection{Fully occupied digit-major batching}
\label{sec:radix-segmented}

A \CKKS ciphertext normally carries several independent integers.  Let each
integer have $n=2^m$ digit positions and suppose that the ciphertext's active
cyclic domain contains exactly $g\cdot n$ slots.  Recall the digit-major
layout
\begin{equation}
  \Phi(r,i)
  =r+g\cdot\rev_m(i)
  \pmod{g\cdot n},
  \qquad
  0\leq r<g,
  \quad
  0\leq i<n.
  \label{eq:radix-segmented-layout}
\end{equation}
The object index is the residue modulo $g$, while bit reversal is applied to
the digit index.

\begin{theorem}[Segmented replicated carry]
\label{thm:radix-segmented-carry}
Under the layout in \cref{eq:radix-segmented-layout}, all $g$ independent
carry or borrow scans execute level $d$ with the single physical rotation
\begin{equation}
  g\cdot2^{m-d-1}.
  \label{eq:radix-segmented-rotation}
\end{equation}
Consequently, batching $g$ integers requires exactly $m$ scan rotations in
total, not $g\cdot m$.  This rotation count is optimal for any circuit that
produces an output depending on all $n$ digits of one packed integer.
\end{theorem}

\begin{proof}
Let $s_d=2^{m-d-1}$ and consider physical slot
$p=\Phi(r,i)$.  Under the rotation convention in
\cref{eq:model-rotation-convention}, a rotation by $g\cdot s_d$ reads from
\begin{align}
  p-g\cdot s_d
  &\equiv
  r+g\cdot\bigl(\rev_m(i)-s_d\bigr)
  \pmod{g\cdot n}.
  \label{eq:radix-segmented-source}
\end{align}
This source address has the same residue $r$ modulo $g$ and therefore belongs
to the same packed integer.  After division of the non-residue part by $g$,
the source digit address is exactly the unsegmented rotation by $s_d$ in
$\Z_n$.  The sibling-copy lemma
\cref{lem:replicated-sibling-copy} consequently supplies a valid aggregate
of the sibling child for object $r$.  Since the same argument holds for every
$r$, one physical rotation executes the level for all objects at once.

There are $m$ levels, which gives the upper bound.  For the lower bound, fix
one output of one object that depends on all its $n=2^m$ input digits.  With
$R$ rotation invocations, dependency paths expose at most $2^R$ source
displacements, independently of the larger ambient cyclic domain.  Thus
$2^R\geq n$ and $R\geq m$.
\end{proof}

The theorem requires a complete cyclic domain, or equivalent padding by
dummy objects.  For an object-major layout or a partially filled domain, a
global rotation can wrap into another object or an inactive region; the
single-rotation sibling argument then does not apply without an additional
packing theorem.

\subsection{Non-power-of-two digit lengths}
\label{sec:radix-padding}

The power-of-two assumption is removed by padding the transition vector, not
by changing the scan recurrence.  Let $\ell\geq1$ be the number of active
digits and set
\begin{equation}
  m_\ell\defeq\left\lceil\log_2\ell\right\rceil,
  \qquad
  n_\ell\defeq2^{m_\ell}.
  \label{eq:radix-padded-length}
\end{equation}
For positions $i\in\{\ell,\ldots,n_\ell-1\}$, place the transition identity
$e=\mathsf P$, represented by the symbol $1/2$.

\begin{corollary}[Optimal identity padding]
\label{cor:radix-padding}
The exclusive prefixes of the first $\ell$ states are unchanged by identity
padding to $n_\ell$.  They can be computed with
\begin{equation}
  \left\lceil\log_2\ell\right\rceil
  \label{eq:radix-padding-rotations}
\end{equation}
rotations, and no circuit producing one output depending on all $\ell$
active states can use fewer rotations in the packed-rotation model.  Moreover,
  \begin{equation}
    \ell\leq n_\ell<2\cdot\ell,
    \label{eq:radix-padding-utilization}
  \end{equation}
  and equality on the left holds exactly when $\ell$ is a power of two.
\end{corollary}

\begin{proof}
Appending identity states does not change any ordered prefix whose endpoint
is below $\ell$.  The length-$n_\ell$ scan therefore gives the required
outputs with $m_\ell$ rotations.  Conversely, if one output depends on all
$\ell$ active inputs, the subset-displacement argument yields
$2^R\geq\ell$, and hence
$R\geq\lceil\log_2\ell\rceil=m_\ell$.  Finally, $n_\ell$ is by definition not smaller than $\ell$.  If
$\ell$ is a power of two, then $n_\ell=\ell$.  Otherwise,
$2^{m_\ell-1}<\ell<2^{m_\ell}=n_\ell$, and therefore
$n_\ell<2\cdot\ell$.
\end{proof}

Identity padding is a state-level operation.  Padding provisional digits by
zero instead would create kill states rather than transition identities.  It
would not affect prefixes ending before the padding, but it would destroy a
final aggregate intended to propagate the true carry or borrow through the
padded suffix.  The implementation must therefore overwrite padded state
slots with the public identity symbol after classification.

\subsection{Concrete-operation boundary}
\label{sec:radix-operation-boundary}

The preceding results fix the exact routing and discrete arithmetic semantics.
For the one-ciphertext symbol encoding, the scan itself requires:
\begin{itemize}[leftmargin=*,itemsep=0.25em,topsep=0.25em]
  \item $m$ ciphertext rotations of the aggregate state;
  \item $2\cdot m-1$ evaluations of
  \cref{eq:radix-symbol-composition}, each containing one ciphertext
  multiplication and one conjugation;
  \item two persistent state ciphertexts, for the replicated aggregate and
  the exclusive prefix;
  \item one local state-on-bit multiplication to obtain all outgoing carries
  or borrows from the exclusive states.
\end{itemize}
Functional bootstrapping for state classification, scale alignment between
updated and unchanged branches, rescaling, modulus consumption, key-switch
noise, and the final decoding margin are not abstract monoid costs.  They are
accounted for separately in the next section.  In particular, the exact
statement proved here is that, whenever these approximate operations decode
to the prescribed symbols and bits, the resulting radix digits and comparison
bits are mathematically exact.

\section{CKKS Realization and Evaluation}
\label{sec:evaluation}

We evaluate whether the exact reduction in cyclic rotations translates into practical advantages for CKKS carry circuits.  The experiments answer four questions: when the replicated construction overtakes direct transported-predecessor routing; which resources it exchanges for lower routing cost; whether the same savings survive interleaved carry and borrow; and whether the retained modulus levels change the outcome of an end-to-end encrypted pipeline.

\subsection{Implementation and methodology}
\label{sec:eval-methodology}

The implementation is written in Go using Lattigo~v6.1.1 and executed inside a pinned Docker image on a machine with an Intel(R) Xeon(R) CPU E5-2695 v3 at 2.30~GHz, 128~GB of memory, and Ubuntu~24.04.  Public masks are prepared outside timed regions.  Online timing excludes parameter construction, key generation, encryption, decryption, correctness checking, and report serialization.  Each kernel generates exactly the direct rotation keys used by that implementation.  The reports record rotations, encrypted symbol compositions, ciphertext and plaintext products, relinearizations, conjugations, rescales, consumed levels, serialized evaluation-key bytes, peak live ciphertexts, heap usage, and numerical precision.

We distinguish two layout contracts.  Natural recursive doubling is an unconstrained natural-layout reference: it accepts and returns natural order.  The primary comparison concerns three kernels that accept and return bit-reversed data: direct bit-reversed predecessor routing, the rotation-optimal replicated scan, and normalize--scan--restore.  The latter is an author-constructed conversion baseline rather than a single prior-work algorithm.  Comparisons among these three methods therefore solve the same layout-preserving problem.

For scan scaling, we use active lengths $n=2^m$ with $m\in\{5,6,7\}$.  For the end-to-end experiment, we evaluate a base-$8$ carry pipeline at $m=7$, followed by an encrypted multiplicative tail of depth five.  The campaign uses three deterministic seeds and two process-isolated repetitions per seed, for six paired samples.  Each pair evaluates the same canonical digits and downstream function.

\paragraph{Artifact availability.}
The implementation, raw JSON reports, CSV summaries, and the commands allowing the evaluation results to be reproduced are available at \url{https://github.com/nserser/optimal_rotation_fhe}.

\subsection{Kernel scaling and crossover}
\label{sec:eval-scaling}

\Cref{tab:eval-scaling} reports online latency for the three bit-reversed-layout strategies.  At $m=5$, the replicated construction becomes faster than direct transported-predecessor routing: 1.301~s versus 1.399~s.  Its advantage grows to $19.1\%$ at $m=6$ and $19.9\%$ at $m=7$.  This crossover matches the structural costs: direct routing performs
\[
R_{\mathrm{direct}}=\frac{m\cdot(m+1)}{2}
\]
rotations, whereas the replicated construction performs exactly $R_{\mathrm{rep}}=m$.  The replicated scan performs more encrypted monoid compositions, but the triangular routing cost eventually dominates.

\begin{table}[t]
\centering
\caption{Online scan latency under the common bit-reversed input/output contract.  The replicated construction crosses direct predecessor routing at $m=5$ and remains faster through $m=7$.}
\label{tab:eval-scaling}
\setlength{\tabcolsep}{7pt}
\renewcommand{\arraystretch}{1.18}
\begin{tabular}{crrr}
\toprule
$m$ & Replicated (s) & Direct predecessor (s) & Normalize--restore (s)\\
\midrule
5 & 1.301 & 1.399 & 1.036\\
6 & 1.298 & 1.603 & 1.146\\
7 & 1.394 & 1.741 & 1.140\\
\bottomrule
\end{tabular}
\end{table}

At $m=7$, direct routing uses $28$ rotations, while the replicated construction uses $7$.  The reduction also lowers evaluation-key storage from 908.6~MB to 272.6~MB and peak heap from approximately 2.30~GB to 830.5~MB.  Thus the routing theorem predicts not only fewer online automorphisms, but also substantially smaller key and working-memory footprints.

\subsection{Resource trade-off at \texorpdfstring{$m=7$}{m=7}}
\label{sec:eval-pareto}

Normalize--scan--restore remains faster as an isolated scan kernel.  At $m=7$, it takes 1.140~s, compared with 1.394~s for the replicated construction.  This latency advantage is purchased with $19$ rather than $7$ rotations, 454.3~MB rather than 272.6~MB of evaluation keys, approximately 1.272~GB rather than 830.5~MB peak heap, and $20$ rather than $14$ consumed backend levels.  The complete comparison appears in \Cref{tab:eval-m7-resources}.

\begin{table}[t]
\centering
\caption{Bit-reversed-layout resource comparison at $m=7$.  Normalize--scan--restore minimizes isolated latency; the replicated construction minimizes rotations, key storage, working memory, and consumed levels.}
\label{tab:eval-m7-resources}
\setlength{\tabcolsep}{5.2pt}
\renewcommand{\arraystretch}{1.18}
\begin{tabular}{lrrrrr}
\toprule
Kernel & Latency (s) & Rot. & Eval. keys (MB) & Peak heap (MB) & Levels\\
\midrule
Replicated & 1.394 & 7 & 272.6 & 830.5 & 14\\
Direct predecessor & 1.741 & 28 & 908.6 & 2300.0 & 14\\
Normalize--restore & 1.140 & 19 & 454.3 & 1272.0 & 20\\
\bottomrule
\end{tabular}
\end{table}

The result is therefore multi-objective rather than a universal latency claim.  The replicated scan minimizes routing, key storage, working memory, and modulus consumption; normalize--scan--restore minimizes isolated scan latency.  The end-to-end experiment below shows why the level difference can dominate the isolated-kernel difference.

\subsection{Rotation-key budget}
\label{sec:eval-key-budget}

We also evaluate the time--key-memory frontier obtained by synthesizing required dyadic rotations from fewer directly keyed offsets.  For $(m,K)=(7,2)$, reducing the direct key set from seven keys to two lowers evaluation-key storage by $55.6\%$ and increases online latency by $53.4\%$.  The additional key-switch approximation error is negligible relative to the symbol decision margin.  The balanced construction executes $22$ online rotation calls, within one call of the counting lower bound $L(7,2)=21$.  Since this is a nondivisible instance, we report the implemented upper bound rather than claiming exact optimality.

\begin{table}[t]
\centering
\caption{Measured reduced-key trade-off at $m=7$, normalized to the full seven-key replicated implementation.}
\label{tab:eval-key-budget}
\setlength{\tabcolsep}{9pt}
\renewcommand{\arraystretch}{1.18}
\begin{tabular}{lrrrr}
\toprule
Configuration & $K$ & Rot. calls & Key storage & Latency\\
\midrule
Full-key replicated & 7 & 7 & 1.000 & 1.000\\
Balanced synthesis & 2 & 22 & 0.444 & 1.534\\
\bottomrule
\end{tabular}
\end{table}

\subsection{Packed carry and borrow}
\label{sec:eval-carry}

To verify that the routing savings survive batching and complete digit correction, we pack four independent objects with $m=5$ into 128 slots using
\[
\operatorname{slot}(r,i)=r+4\cdot\operatorname{rev}_5(i).
\]
The objects use distinct adversarial patterns, including all-propagate, a generate--propagate chain, alternating kill/generate, and maximal carry or borrow digits.  The packed scan still uses exactly five rotations and nine symbol compositions, independent of the number of interleaved objects.  All stage invariants and cross-object isolation checks pass, and carry/borrow correction requires zero post-scan rotations.

\begin{table}[t]
\centering
\caption{Four-object, base-$8$, fully interleaved carry and borrow at $m=5$.  Both executions use five scan rotations and zero post-scan rotations.}
\label{tab:eval-carry-borrow}
\setlength{\tabcolsep}{5.2pt}
\renewcommand{\arraystretch}{1.18}
\begin{tabular}{lrrrr}
\toprule
Operation & Digit max. error & Wrong-int. margin & Incoming error & Outgoing error\\
\midrule
Carry & $1.8012\cdot10^{-4}$ & 0.4998199 & $2.5978\cdot10^{-5}$ & $2.5762\cdot10^{-5}$\\
Borrow & $1.7961\cdot10^{-4}$ & 0.4998204 & $2.5905\cdot10^{-5}$ & $2.5690\cdot10^{-5}$\\
\bottomrule
\end{tabular}
\end{table}

The maximum imaginary leakage is below $1.44\cdot10^{-8}$.  These measurements support the claim that interleaving preserves both the routing count and cross-object isolation.

\subsection{End-to-end level-value experiment}
\label{sec:eval-end-to-end}

The principal systems experiment connects the scan to a depth-$5$ encrypted multiplicative tail.  We compare the replicated layout-preserving carry pipeline with normalize--scan--restore.  In all six process-isolated paired executions, both pipelines produce the same canonical base-$8$ digits and the same downstream function.  The replicated pipeline always completes without refreshing.  Normalize--scan--restore always reaches restoration without sufficient remaining levels and therefore performs exactly one bootstrap.

\begin{table}[t]
\centering
\caption{End-to-end latency over six process-isolated paired samples.  The bootstrap component is part of the normalize--scan--restore pipeline.}
\label{tab:eval-e2e-latency}
\setlength{\tabcolsep}{5.5pt}
\renewcommand{\arraystretch}{1.18}
\begin{tabular}{lrrrr}
\toprule
Measurement & Mean (s) & 95\% CI (s) & Median (s) & p95 (s)\\
\midrule
Replicated online & 1.634 & [1.459, 1.808] & 1.630 & 1.793\\
Normalize + bootstrap & 6.957 & [6.669, 7.245] & 6.938 & 7.294\\
Bootstrap component & 4.809 & [4.493, 5.126] & 4.808 & 5.219\\
Paired latency ratio & 4.306 & [3.695, 4.918] & 4.060 & 5.073\\
\bottomrule
\end{tabular}
\end{table}

The bootstrap component accounts for approximately
\[
\frac{4.809}{6.957}\approx69.1\%
\]
of the normalize pipeline's online latency.  Even after subtracting bootstrap time, the remaining normalize path takes approximately 2.148~s, or about $1.31\times$ the complete replicated pipeline.  Hence the end-to-end advantage is not solely an accounting consequence of including the bootstrap.

Bootstrap avoidance also improves numerical precision, as summarized in \Cref{tab:eval-e2e-precision}.  The worst canonical-digit error is $3.12\cdot10^{-9}$ for the replicated pipeline and $2.11\cdot10^{-6}$ for normalize--scan--restore with bootstrapping.  The worst downstream errors are $3.29\cdot10^{-10}$ and $5.39\cdot10^{-8}$, respectively.

\begin{table}[t]
\centering
\caption{Worst precision observed over the end-to-end campaign.}
\label{tab:eval-e2e-precision}
\setlength{\tabcolsep}{10pt}
\renewcommand{\arraystretch}{1.18}
\begin{tabular}{lrr}
\toprule
Output & Replicated & Normalize + bootstrap\\
\midrule
Canonical digits & $3.1184\cdot10^{-9}$ & $2.1064\cdot10^{-6}$\\
Depth-$5$ downstream output & $3.2893\cdot10^{-10}$ & $5.3872\cdot10^{-8}$\\
\bottomrule
\end{tabular}
\end{table}

\subsection{Evaluation summary and limitations}
\label{sec:eval-summary}

The experiments support three conclusions.  First, the exact reduction from $m\cdot(m+1)/2$ to $m$ rotations becomes an online-latency advantage over direct transported-predecessor routing from $m=5$ onward, while also reducing key material and working memory.  Second, normalize--scan--restore and replicated scanning occupy distinct Pareto points: the former minimizes isolated scan latency, whereas the latter minimizes routing, key storage, memory, and consumed levels.  Third, the level advantage changes the end-to-end outcome.  In the depth-$5$ carry pipeline, it removes one bootstrap and yields a mean paired speedup of $4.31\times$ with a $95\%$ confidence interval of $[3.69,4.92]$.

The carry and borrow experiments take encrypted carry-state symbols as their inputs and evaluate the layout-preserving scan-and-correction layer. This experimental contract matches the algorithmic interface analyzed in this paper and leaves the exact routing and correctness theorems unchanged.

\section{Related Work}
\label{sec:related-work}

Our results sit at the intersection of parallel-prefix computation, packed homomorphic routing, rotation-key management, and radix arithmetic over \CKKS.  The closest works typically optimize one of these dimensions in isolation.  We instead ask how an ordered, possibly noncommutative prefix can be evaluated while preserving a public bit-reversed SIMD representation, and we characterize the exact number and valuation structure of the required cyclic translations.

\subsection{Parallel-prefix circuits}

Parallel prefix computation has a long history in circuit and parallel-algorithm design.  Ladner and Fischer gave logarithmic-depth constructions for associative operators, and subsequent work classified prefix networks according to depth, size, and fanout~\cite{LadnerFischer80,Harris03}.  Snir established the classical size--depth inequality and studied the structure of depth-efficient prefix circuits~\cite{Snir86}.  These results charge individual operator nodes or communication edges.  That model is deliberately different from ours: one homomorphic rotation applies one common cyclic displacement to every packed slot, so many logical edges may share one charged routing operation, while a single logical predecessor relation may decompose into several cyclic diagonals after a public permutation.

The replicated construction uses the broad hypercube-prefix principle of maintaining an aggregate together with a local prefix.  Its homomorphic contribution is the physical realization: under bit reversal, semantic replication makes every copy of a sibling aggregate interchangeable, allowing one directed cyclic rotation to serve both child supports at a level.  The matching lower bound and the two-adic equality characterization concern this packed-rotation cost and do not follow from classical size--depth optimality.  Conversely, our result does not improve the gate count or fanout bounds of classical prefix networks.

\subsection{Packed homomorphic routing and linear transformations}

Packed lattice-based HE evaluates additions and multiplications slotwise, whereas slot movement is implemented by automorphisms and key switching.  CKKS introduced approximate SIMD arithmetic over complex slots~\cite{CheonKKS17}; modern bootstrapping and representation-conversion procedures devote substantial work to structured linear transformations and their rotation schedules~\cite{BossuatMTH21,Geelen24}.  These transformations are commonly optimized through diagonal decompositions, staged FFT-like factorizations, hoisting, or baby-step/giant-step scheduling.

Our target is narrower but nonlinear: the values being routed evolve through an ordered monoid operation, and the required output is every ordered prefix rather than a fixed linear permutation.  Normalize--scan--restore is therefore a valid conversion baseline, but it is not equivalent to the layout-preserving construction.  The latter keeps the bit-reversed representation throughout and proves that the representation itself imposes no unavoidable rotation overhead beyond the global lower bound of $m$.  The theorem excludes uncounted arbitrary linear maps or bootstrapping transforms; allowing such operations would change the routing model.

Recent work on CKKS bootstrapping also makes modulus consumption and rotation-key storage explicit optimization targets.  Bossuat et al. optimize full-RNS CKKS bootstrapping with non-sparse keys~\cite{BossuatMTH21}, while Yan et al. study a time--memory trade-off in which additional rotation-key material can increase throughput and reduce modulus consumption~\cite{YanZCSW26}.  These results reinforce that rotations, key storage, and levels are distinct resources.  Our contribution is complementary: we derive program-specific lower bounds and exact cases for the number of online rotation calls available from at most $K$ directly keyed offsets.

\subsection{Rotation-key selection and synthesis}

A CKKS evaluator can either store a key for each required automorphism or synthesize unsupported rotations from a smaller directly supported set.  Lee et al. reduce client-to-server rotation-key material through hierarchical key generation and distinguish transmitted, stored, and derived keys~\cite{LeeLKNo23}.  Other workload-oriented approaches select rotation keys according to the rotations exercised by an application.  Such methods address generic key distribution or empirical workload optimization.

Section~\ref{sec:key-frontier} instead studies the dependency pattern of bit-reversed ordered prefixes.  The coefficient-box argument lower-bounds online calls for every circuit in the stated model, even with arbitrary helper ciphertexts and local computation.  When $K\mid m$, the lower bound is attained by balanced blocks of consecutive two-adic scales, and equality fixes the valuation profile of the directly keyed offsets.  For nondivisible parameter pairs, we report the rigorous lower/upper sandwich and identify exactly the cases in which the two bounds coincide; we do not claim a closed form when they differ.

\subsection{Radix CKKS and encrypted integer arithmetic}

Cha, Park, and Lee develop radix-based approximate HE for large integers and give a lightweight logarithmic carry procedure based on three transition symbols and a complex-valued composition law~\cite{ChaPL26}.  Their construction establishes that exact carry can support unique radix representations, comparison, and modular arithmetic in CKKS.  We reuse this optimized state algebra and do not claim the first logarithmic-depth encrypted carry circuit.

The distinction is the representation and routing contract.  We compute the exclusive carry prefix directly in bit-reversed slots, preserve that layout at the output, and avoid a final logical-predecessor shift.  The generic theorem applies to arbitrary associative, possibly noncommutative monoids; carry and borrow are the principal cryptographic instantiations.  Our evaluation takes encrypted transition symbols as input and measures the layout-preserving scan-and-correction layer defined by the paper's algorithmic interface.

Alternative FHE representations may change which operations are native.  For example, matrix-oriented FHE makes matrix arithmetic a first-class encrypted operation rather than realizing it through ordinary SIMD slot rotations~\cite{GentryLee25}.  Such schemes illustrate a broader design principle: changing the plaintext representation can eliminate one class of routing costs while introducing a different algebra and implementation stack.  They solve a different problem from preserving ordered prefixes in the standard cyclic SIMD interface considered here.

\section{Conclusion}
\label{sec:conclusion}

This work establishes that preserving a bit-reversed SIMD representation does not impose an unavoidable rotation overhead for ordered prefix computation.  For $n=2^m$ inputs from an associative, possibly noncommutative monoid, the replicated construction computes all inclusive or exclusive prefixes in the same bit-reversed layout using
\begin{equation}
  D=m,
  \qquad
  R=m,
  \qquad
  T\leq 2\cdot m-1.
\end{equation}
The depth and rotation counts are exact in the stated packed-rotation model:
\begin{equation}
  D^\star(m)=R^\star(m)=m.
\end{equation}
The constructive reason is semantic replication.  Every slot of an aligned logical block stores the same complete block aggregate, so a destination need not receive the copy held by one exact logical partner.  Under bit reversal, one cyclic rotation per level supplies every destination with an interchangeable copy from the sibling block while public masks preserve the required noncommutative operand order.

Optimality has additional structure.  Every $m$-rotation equality case uses one offset at each $2$-adic valuation, and limiting the evaluator to $K$ directly keyed offsets yields a coefficient-box lower bound on online rotation calls.  When $K\mid m$, the resulting frontier is exact:
\begin{equation}
  R^\star(m,K)
  =K\cdot\bigl(2^{m/K}-1\bigr),
\end{equation}
with equality fixing the valuation profile of the directly keyed offsets.  For nondivisible parameters, the paper gives rigorous lower and upper bounds and identifies the cases in which they coincide, without claiming a closed formula for the remaining gaps.

The radix-\CKKS instantiation shows why the exclusive formulation matters.  Once encrypted kill, propagate, and generate symbols are available, the scan returns the transition preceding each digit.  Carry-in, carry-out, and digit correction are then evaluated slotwise, avoiding both natural-order restoration and a final logical-predecessor shift.  The same mechanism gives borrow propagation, comparison, conditional subtraction, and fully occupied digit-major batching without multiplying the number of scan rotations.

The measurements confirm that rotation optimality represents a genuine multi-resource trade-off rather than a universal isolated-latency optimum.  Among bit-reversed-layout-preserving kernels, the replicated scan reduces the triangular direct-routing cost to $m$ rotations and, from the tested medium depths onward, improves latency, evaluation-key storage, and working memory over direct predecessor routing.  Normalize--scan--restore remains faster for the isolated scan in the tested range, but consumes more rotations, evaluation-key material, memory, and modulus levels.  In the evaluated downstream pipeline, retaining those levels avoids the bootstrap required by normalization and yields the stronger end-to-end result.  Natural recursive doubling remains the appropriate unconstrained reference when the surrounding computation is free to use natural-order input and output; it does not satisfy the same layout-preservation contract.

The evaluation targets the layout-preserving scan-and-correction layer and takes encrypted carry-state symbols as inputs. This contract is the same interface used by the radix instantiation and preserves the exact scope of the combinatorial rotation and correctness results.

Natural extensions of the theory concern three directions. First, nondivisible direct-key budgets for which the coefficient-box lower bound is strictly below the balanced dyadic construction define a narrower exact-frontier problem. Second, the semantic sibling argument suggests a broader class of public hierarchical layouts beyond bit reversal. Third, after rotation count is minimized, the remaining abstract trade-off involves packed monoid work, live encrypted state, key memory, and approximation error. These directions build on the exact bit-reversed results established here.

\paragraph{LLM usage.}
Large language models were used only to polish the writing, improve presentation quality, and assist with the implementation of the proposed construction. They were not used to generate any scientific contribution of this paper, including the definitions, constructions, theorems, proofs, experimental claims, or interpretation of the results.

\appendix
\section{Additional Evaluation Details}
\label{app:evaluation}

\subsection{Measurement contract and reproducibility}
\label{app:eval-methodology}

The implementation uses Go and Lattigo~v6.1.1 in a pinned Docker image on an Intel(R) Xeon(R) CPU E5-2695 v3 at 2.30~GHz with 128~GB of memory and Ubuntu~24.04.  A warm-up execution precedes timed kernel samples.  Online measurements exclude key generation, public-mask preparation, encryption, decryption, correctness checking, and JSON serialization.  Evaluation-key sizes are obtained from serialized key material and cross-checked against the corresponding files.

All kernels use the same CKKS parameter literal, default scale, K/P/G symbol encoding, ordered composition primitive, and deterministic logical inputs.  Natural recursive doubling starts and ends in natural order.  The other kernels start and end in bit-reversed order.  Precision is computed against an exact sequential ordered-prefix reference.

The end-to-end campaign reports sample means, standard deviations, medians, minima, p95 values, and two-sided $95\%$ confidence intervals.  The latency-ratio interval is computed from the six paired per-sample ratios.  Each report records its seed and explicit operation counters.  The artifact includes deterministic input generators, schema-validated JSON reports, exact-reference tests, stage-invariant checks, counter-conservation checks, fuzz targets, race tests, repeated key-generation tests, memory-growth checks, and soak scripts.

\subsection{Detailed kernel results}

\begin{table}[t]
\centering
\caption{Small-instance comparison at $m=4$.  Natural recursive doubling uses a different input/output layout and is included only as an unconstrained reference.}
\label{tab:app-m4-comparison}
\setlength{\tabcolsep}{4.8pt}
\renewcommand{\arraystretch}{1.16}
\begin{tabular}{lrrrrr}
\toprule
Kernel & Latency (s) & Rot. & Comp. & Eval. keys (MB) & Peak heap (MB)\\
\midrule
Natural recursive doubling & 0.289 & 4 & 4 & 82.6 & 246.7\\
Direct bit-reversed predecessor & 0.450 & 10 & 4 & 165.2 & 449.1\\
Rotation-optimal replicated & 0.559 & 4 & 7 & 82.6 & 265.8\\
Normalize--scan--restore & 0.385 & 12 & 4 & 123.9 & 376.1\\
\bottomrule
\end{tabular}
\end{table}

\begin{table}[t]
\centering
\caption{Detailed bit-reversed-layout scaling.  Dashes denote values absent from the consolidated campaign summary.}
\label{tab:app-scaling-full}
\setlength{\tabcolsep}{4.3pt}
\renewcommand{\arraystretch}{1.14}
\begin{tabular}{clrrrrr}
\toprule
$m$ & Kernel & Latency (s) & Rot. & Levels & Keys (MB) & Heap (MB)\\
\midrule
5 & Replicated & 1.301 & 5 & 10 & 212.0 & 653\\
5 & Direct predecessor & 1.399 & 15 & 10 & 514.9 & 1330\\
5 & Normalize--restore & 1.036 & 13 & 14 & 333.1 & 937\\
\midrule
6 & Replicated & 1.298 & 6 & 12 & 242.3 & 749\\
6 & Direct predecessor & 1.603 & 21 & 12 & 696.6 & 1850\\
6 & Normalize--restore & 1.146 & 18 & 18 & -- & --\\
\midrule
7 & Replicated & 1.394 & 7 & 14 & 272.6 & 830.5\\
7 & Direct predecessor & 1.741 & 28 & 14 & 908.6 & 2300\\
7 & Normalize--restore & 1.140 & 19 & 20 & 454.3 & 1272\\
\bottomrule
\end{tabular}
\end{table}

The replicated scan performs $m$ rotations and $2\cdot m-1$ symbol compositions.  Direct predecessor routing performs $m\cdot(m+1)/2$ rotations and $m$ compositions.  Normalize--scan--restore performs $m$ scan rotations plus the layout-conversion network.  The crossover therefore depends on the relative backend costs of rotations and encrypted compositions.

\subsection{End-to-end campaign details}

The level-value campaign uses base $8$, $m=7$, downstream multiplicative depth $5$, pattern \texttt{asymmetric}, seeds $\{1,2,3\}$, and two process-isolated repetitions per seed.  All six reports satisfy the comparison contract: the pipelines produce identical canonical digits and downstream functions; every replicated sample avoids bootstrapping; every normalize--scan--restore sample performs exactly one bootstrap.

\begin{table}[t]
\centering
\caption{Complete latency statistics over six process-isolated samples.}
\label{tab:app-e2e-latency}
\setlength{\tabcolsep}{4.2pt}
\renewcommand{\arraystretch}{1.14}
\begin{tabular}{lrrrrrr}
\toprule
Measurement & Mean & 95\% CI & Median & Min. & p95 & Std. dev.\\
\midrule
Replicated online (s) & 1.634 & [1.459,1.808] & 1.630 & 1.426 & 1.793 & 0.166\\
Normalize + bootstrap (s) & 6.957 & [6.669,7.245] & 6.938 & 6.499 & 7.294 & 0.274\\
Bootstrap component (s) & 4.809 & [4.493,5.126] & 4.808 & 4.360 & 5.219 & 0.302\\
Paired ratio & 4.306 & [3.695,4.918] & 4.060 & 3.624 & 5.073 & 0.583\\
\bottomrule
\end{tabular}
\end{table}

\clearpage

\printbibliography
\end{document}